\pdfoutput=1

\documentclass[manuscript,screen]{acmart}
\usepackage[utf8]{inputenc}
\usepackage{multirow}
\usepackage{amsmath}
\usepackage{subcaption}

\theoremstyle{definition}

\newtheorem{theorem}{Theorem}

\theoremstyle{plain}
\newtheorem{proposition}{Proposition}[section]

\theoremstyle{remark}
\newtheorem{proof_outl}{Proof Outline}[section]

\title{Maximal fairness}
\author{MaryBeth Defrance}
\affiliation{%
  \institution{University of Ghent}
  \city{Ghent}
  \country{Belgium}
  \postcode{9000}
}
\email{marybeth.defrance@ugent.be}
\author{Tijl De Bie}
\affiliation{%
  \institution{University of Ghent}
  \city{Ghent}
  \country{Belgium}
  \postcode{9000}
}
\email{tijl.debie@ugent.be}
\renewcommand{\shortauthors}{Defrance and De Bie}

\begin{document}

\begin{abstract}
Fairness in AI has garnered quite some attention in research, and increasingly also in society. The so-called "Impossibility Theorem" has been one of the more striking research results with both theoretical and practical consequences, as it states that satisfying a certain combination of fairness measures is impossible. To date, this negative result has not yet been complemented with a positive one: a characterization of which combinations of fairness notions \emph{are} possible. This work aims to fill this gap by identifying maximal sets of commonly used fairness measures that can be simultaneously satisfied. The fairness measures used are demographic parity, equal opportunity, false positive parity, predictive parity, predictive equality, overall accuracy equality and treatment equality. 
We conclude that in total 12 maximal sets of these fairness measures are possible, among which seven combinations of two measures, and five combinations of three measures.
Our work raises interest questions regarding the practical relevance of each of these 12 \emph{maximal fairness} notions in various scenarios.
\end{abstract}



\maketitle

\section{Introduction}

The ProPublica article about the risk assessment system COMPAS generated quite some buzz in the field of fairness of AI \cite{COMPAS_article}. It started raising questions on what characteristics an AI system needs to satisfy for it to be fair in a practical, intuitive, or legal sense. One of the main criticisms from the authors of the article were the differing false positive and false negative rates for black people as compared to for white people. In response, the creators of the COMPAS-tool defended their system by arguing that it was properly calibrated, meaning that the predicted probability of recidivism was accurate for both demographic groups, and that this made the system fair. One might argue that ideally these properties would both be satisfied. However, the works of Chouldechova et al. \cite{disparate_impact_fairness}, Hardt et al. \cite{Equal_opportunity}, and Kleinberg et al. \cite{trade_offs_fairness_risk_scores} all led to the same conclusion, colloquially referred to as the "Impossibility Theorem". It states that it is mathematically \emph{impossible} to have calibration, equal false positive and equal false negative rates all at the same time, except in a practically irrelevant degenerate case. (See Sec.~\ref{sec:related} for some further details.)

\paragraph{Contributions}
To the best of our knowledge, this negative result has never been complemented with a positive one, namely the characterisation of maximal sets of fairness notions that can be simultaneously satisfied. Thus, in this paper, we investigate which and how many of the more frequently used fairness measures can be combined, leading to a set of maximal notions of fairness. We do this for the simplest but commonly studied case of binary classification.
Out of a range of seven commonly used fairness measures, we identify a total of 12 maximal combinations, including seven maximal combinations of two and five maximal combinations of three fairness measures.
Our results trigger intriguing questions on the practical relevance and interpretation of each of these maximal fairness notions.

\section{Related Work}\label{sec:related}

Most of the conceptual research about fairness in machine learning is about the introduction of new fairness measures, about the interpretation of previously proposed fairness notions, about techniques to satisfy fairness measures, or about so-called impossibility theorems that show that certain combinations of fairness measures are impossible to achieve. Here we survey the work most directly related to the present paper.

\subsection{Surveys on Fairness definitions}

Verma S. and Rubin J. \cite{Fairness_definitions_explained} discuss different fairness measures for classification with a focus on creating a human-understandable interpretation of them and applying them to a specific use case. A more in-depth discussion about the design of a predictive model concerning pitfalls and biases that occur in the process can be found in Mitchell et al.  \cite{Algorithmic_fairness_choices_ass_def}. This paper also covers different fairness measures and relates the choice between these definitions back to the discussion of design choices, and it lightly touches on the impossibilities of combining certain fairness measures when working with scores instead of binary decisions. The survey of N. Mehrabi et al. contains an extensive list of possible biases and fairness measures \cite{large_bias_survey}. The fairness measures are grouped by type, group, subgroup or individual. It also discusses methods to satisfy them. Ruf and Detyniecki \cite{fairness_compass} strive to create a type of Fairness Compass which would help practitioners decide what fairness measure is best suited in a given problem. Throughout the work, different fairness measures are clearly explained through the use of confusion matrices and with a theoretical use case. 

\subsection{Prior work on combining Fairness definitions}

Most research related to combining fairness measures focusses on the impossibility of certain combinations. 
Chouldechova A. \cite{disparate_impact_fairness} discuss how adherence to a certain criterion can still lead to a considerable disparate impact between groups. It states that it is impossible to achieve equality between groups for false positive rate (FPR), positive predictive value (PPV) and false negative rate (FNR) when the base rates differ between groups. Kleinberg et al.  \cite{trade_offs_fairness_risk_scores} discuss combining three different fairness measures in the context of risk scoring. It is focussed on scoring instead of binary classification. They come to the conclusion that combining those three measures is only possible under unique circumstances, namely if the groups have equal base rates or if the model is capable of perfect prediction. 
Berk et al. \cite{Extensive_trade_offs_fairness} focus on the trade-off that would occur when using multiple incompatible fairness measures such as discussed in the work of Chouldechova A. \cite{disparate_impact_fairness}. They also include different techniques that can be used to achieve this goal.

\section{Maximal Combinations of Fairness Definitions}

Common fairness measures can be expressed by referring to the confusion matrices (as in Table~\ref{tab:ex_conf_matrix}) for the different protected demographic groups. They are defined by choosing a statistic computed on a confusion matrix, for example the false negative rate or some other statistic, and requiring this statistic to be equal across the demographic groups. We will demonstrate this for seven common fairness measures in Sec.~\ref{sec:fairnessdefs}. This approach makes combining fairness metrics easy: simply by combining the constraints they impose on the contingency matrices.
In this paper we assume there are two demographic groups, and we will refer to the statistics computed on the respective confusion matrix using subscripts a and b. Note that in their most common definition, confusion matrices contain counts. However, because this simplifies our derivations somewhat, in this paper we normalize them by dividing each of the four cells by their sum.

\begin{table}[h]
\begin{center}
\caption{A confusion matrix, showing the relation between ground truth labels, predicted labels, true positives, true negatives, false positives, and false negatives. Commonly used fairness notions require a particular statistic computed on the confusion matrices for different demographic groups to be equal to each other.}
\label{tab:ex_conf_matrix}
\begin{tabular}{|l|c|c|c|}
\cline{3-4}
 \multicolumn{2}{c|}{}  & \multicolumn{2}{c|}{Predicted}  \\ \cline{3-4}
  \multicolumn{2}{c|}{}  & Positive & Negative \\ \hline
 \multirow{2}{*}{True} & Positive & TP & FN \\ \cline{2-4}
 & Negative & FP & TN \\ \hline
\end{tabular}
\end{center}
\end{table}

By definition, the (normalized) confusion matrices satisfy certain constraints of their own. These constraints are expressed in the system of equations depicted in Figure~\ref{eq:constraints}. Here, the variables $1-x$ and $1-y$ denote the base rates of group a and group b respectively. The base rate is the fraction of positive samples in the data set for a given group. The inequalities will be left out of subsequent derivations in order to simplify the notations.

\begin{figure}
\[ \begin{cases}
	TP_a + FN_a = 1 - x \\
	TN_a + FP_a = x \\
	TP_b + FN_b = 1 - y \\
	TN_b + FP_b = y \\
	0 <= TP_a <= 1-x, 0 <= TP_b <= 1-y  \\
	0 <= TN_a <= x, 0 <= TN_b <= y \\
	0 <= FP_a <= x, 0 <= FP_b <= y \\
	0 <= FN_a <= 1-x, 0 <= FN_b <= 1-y \\ 
\end{cases}  \]
\caption{Constraints always present in the system of equations, due to the properties of the confusion matrix.}
\Description{The constraints that come from the properties of the confusion matrix}
\label{eq:constraints}
\end{figure}

Before investigating which fairness measures can be combined, we must be precise about what we mean by that.
To do this, let us first point out that any combination of fairness measures can be satisfied if the base rates are equal between both demographic groups, i.e. if $x=y$. Indeed, in those cases it is possible for the matrices to be exact copies of each other. However as base rates are a property of the data and hence cannot be controlled, this possibility must not lead us to judge that any set of fairness definitions is compatible. Another possibility we will see in our derivations below is that certain combinations are possible only if some of the variables in the confusion table are equal to a fixed (and often trivial) value. Again, although this constraint can be satisfied in principle, it leaves very little freedom to the classifier, effectively preventing it from being useful.
Thus, we argue that also such combinations cannot be deemed `possible'.
The following definition determines what properties are necessary to deem a combination of fairness measures as `possible'.

\begin{definition}[Possibility of combining fairness measures] \label{def:possible}
A combination of fairness measures is deemed \emph{possible} if after combining the constraints, all elements in both confusion matrices can still take on values within a subset of $[0,1]$ with non-zero Lebesgue measure, and this for all pairs of base rates $(x,y)$ from a subset of $[0,1]^2$ with non-zero 2-dimensional Lebesgue measure. Or informally speaking, the elements of the confusion matrices should be able to take on values within a non-trivial range, and this for all base rate pairs from within a non-trivial 2-dimensional range.
\end{definition}

\subsection{Implemented Fairness measures}\label{sec:fairnessdefs}

The fairness measures used in this paper are demographic parity, equal opportunity, false positive parity, predictive parity, predictive equality, overall accuracy equality and treatment equality. More detailed information about these measures can be found in Table~\ref{tab:all_fairness_def}. This table includes two equivalent definitions: one expressed in terms of probabilities, and another in terms of a constraint on the confusion matrices. The table also mentions the statistical property it requires to be equal between groups (if it has a standard name). Finally, it mentions the `orientation' of the constraint on the confusion matrices, which may offer insight into the possible combinations of fairness notions. Other frequently used fairness measures exist, particularly in other settings than binary classification, but in the paper our focus is on measures for this setting only. Note that each of these measures only require one constraint to be satisfied, and by combining them the number of constraints is growing.

Below we discuss each of the considered fairness measures in some greater detail. Figure~\ref{fig:orient_defs} also illustrates them more visually, highlighting the orientations on the confusion matrix, and in this way also showing their differences and similarities.

\begin{table}
\begin{center}
\caption{The list of fairness measures considered in this paper, with their definition expressed in probabilities (second column) as well as by means of an expression evaluated on the confusion matrices (third column), and the statistical property it requires to be equal (fourth column). The fourth column mentions the orientation of the constraints on the confusion matrix.}
\label{tab:all_fairness_def}
\begin{tabular}{|l|c|c|c|c|}
\hline
Name & In probabilities & In the confusion matrix & Statistic & Orientation \\ \hline
\multirow[c]{3}{*}{Demographic Parity} & \multirow[l]{3}{*}{$P(\hat{Y} | A=0) = P(\hat{Y} | A=1)$} & \multirow[c]{3}{2.8cm}{$\frac{FP_A+TP_A}{FP_A+TP_A+FN_A+TN_A} = \frac{FP_B+TP_B}{FP_B+TP_B+FN_B+TN_B}$}  & \multirow[c]{3}{2.6cm}{\centering Positive Rate  (PR)} & \multirow[c]{2}{*}{$B_v$ Board} \\
 & &  & & \\ 
 & &  & &  \\ \hline
\multirow[c]{2}{*}{Equal Opportunity} & \multirow[l]{2}{3.3cm}{$P(\hat{Y} = 1 | A = 0, Y=1) = P(\hat{Y} = 1 | A = 1, Y=1)$} & \multirow[c]{2}{*}{$\frac{TP_A}{TP_A+FN_A} = \frac{TP_B}{TP_B+FN_B}$}  & \multirow[c]{2}{2.6cm}{\centering False Negative Rate (FNR)} & H \\
 & &  & & Horizontal \\ \hline
\multirow[c]{2}{*}{False positive parity} & \multirow[l]{2}{3.3cm}{$P(\hat{Y} = 1 | A = 0, Y=0) = P(\hat{Y} = 1 | A = 1, Y=0)$} & \multirow[c]{2}{*}{$\frac{FP_A}{FP_A+TN_A} = \frac{FP_B}{FP_B+TN_B}$} & \multirow[c]{2}{2.6cm}{\centering False Positive Rate (FPR)} & H \\
 & &  & & Horizontal\\ \hline
\multirow[c]{2}{*}{Predictive Parity} & \multirow[l]{2}{3.3cm}{$P(Y=1 | \hat{Y}=1, A=0)=P(Y=1 | \hat{Y}=1, A=1)$} & \multirow[c]{2}{*}{$\frac{TP_A}{TP_A+FP_A} = \frac{TP_B}{TP_B+FP_B}$}  & \multirow[c]{2}{2.6cm}{\centering Positive Prediction Value (PPV)} & V \\ 
 & &  &  & Vertical \\ \hline
\multirow[c]{2}{*}{Predictive Equality} & \multirow[l]{2}{3.3cm}{$P(Y=1 | \hat{Y} = 0, A=0) = P(Y=1 | \hat{Y} = 0, A=1)$} & \multirow[c]{2}{*}{$\frac{FN_A}{FN_A + TN_A} = \frac{FN_B}{FN_B + TN_B}$}  & \multirow[c]{2}{2.6cm}{\centering False Omission Rate (FOR)} & V \\ 
 & &  & & Vertical \\ \hline
\multirow[l]{3}{2.7cm}{Overall accuracy Equality} & \multirow[l]{3}{3cm}{$P(\hat{Y} = Y | A=0) = P(\hat{Y} = Y | A=1) $} & \multirow[c]{3}{2.8cm}{$\frac{TP_A + TN_A}{TP_A+FP_A+TN_A+FN_A} = \frac{TP_B + TN_B}{TP_B+FP_B+TN_B+FN_B} $}  & \multirow[c]{3}{2.6cm}{\centering Accuracy (ACC)} & \multirow[c]{3}{*}{$B_d$ Board}\\ 
 & &  & & \\
 & &  & & \\ \hline
\multirow[c]{2}{*}{Treatment Equality} & \multirow[c]{2}{*}{$ - $} & \multirow[c]{2}{*}{$ \frac{FN_A}{FP_A} = \frac{FN_B}{FP_B}$}  & \multirow[c]{2}{*}{} & D\\  
& &  & & Diagonal\\ \hline

\end{tabular}
\end{center}
\end{table}

\begin{definition}[Demographic Parity \cite{Dem_parity}] Demographic parity is a fairness measure which requires that both groups proportionally receive the positive outcome (TP + FP) equally regardless of the ground truth and thus also regardless of the base rates. This independence makes demographic parity unusual. Because demographic parity is independent of the ground truth it often tends to result in lower accuracies when enforced. Note that a similar notion is conditional statistical parity \cite{Predictive_equality}, which is similar except that it requires equality between sensitive groups when restricted to people that share a certain attribute. Conditional statistical parity will not be discussed in this paper.
\end{definition}
\begin{definition}[Equal Opportunity \cite{Equal_opportunity}] Equal opportunity is both a measure of itself and also one of two conditions that is required for equalised odds. Equal opportunity means that people have the same probability regardless of their sensitive attribute of receiving a negative prediction given that they belong in the positive category.
\end{definition} 
\begin{definition}[False positive parity \cite{Equal_opportunity}] False positive parity is the second condition for equalised odds, alongside equal opportunity. It requires that the probability of receiving a positive prediction when the actual class is negative is equal across sensitive groups. If both false positive parity and equal opportunity are satisfied then equalised odds is satisfied.
\end{definition}
\begin{definition}[Predictive Parity \cite{disparate_impact_fairness}] Predictive parity differs from equal opportunity in that it is concerned with the probabilities of the actual class and not of the predicted class. In other words, rather than horizontal ratios in the confusion matrices for equal opportunity and false positive parity, predictive parity considers a vertical ratio. In order for predictive parity to be satisfied, the probability that a positive prediction is correct must be equal across sensitive groups. 
\end{definition}
\begin{definition}[Predictive Equality \cite{Predictive_equality}] Predictive equality is very analogous to predictive parity.  It requires that the probability for a negative prediction to be correct is equal across sensitive groups. The combination of satisfying both predictive parity and predictive equality is also called conditional use accuracy equality. 
\end{definition}
\begin{definition}[Overall accuracy Equality \cite{Extensive_trade_offs_fairness}] Overall accuracy equality requires that the accuracy of the classifier is equal between sensitive groups. While straightforward and intuitive, this definition on its own is often not sufficient to guarantee intuitive fairness. Indeed, it implicitly assumes an equal impact of positive and negative predictions, which is often unjustified in practice.
\end{definition}
\begin{definition}[Treatment Equality \cite{Extensive_trade_offs_fairness}] Treatment equality is an atypical fairness measure as it cannot be expressed as a probability nor a commonly used statistical property. It is a constraint on the false positives and false negatives, requiring their ratio to be equal between sensitive groups.  
\end{definition}

\begin{figure}
\includegraphics[width = 0.75\linewidth]{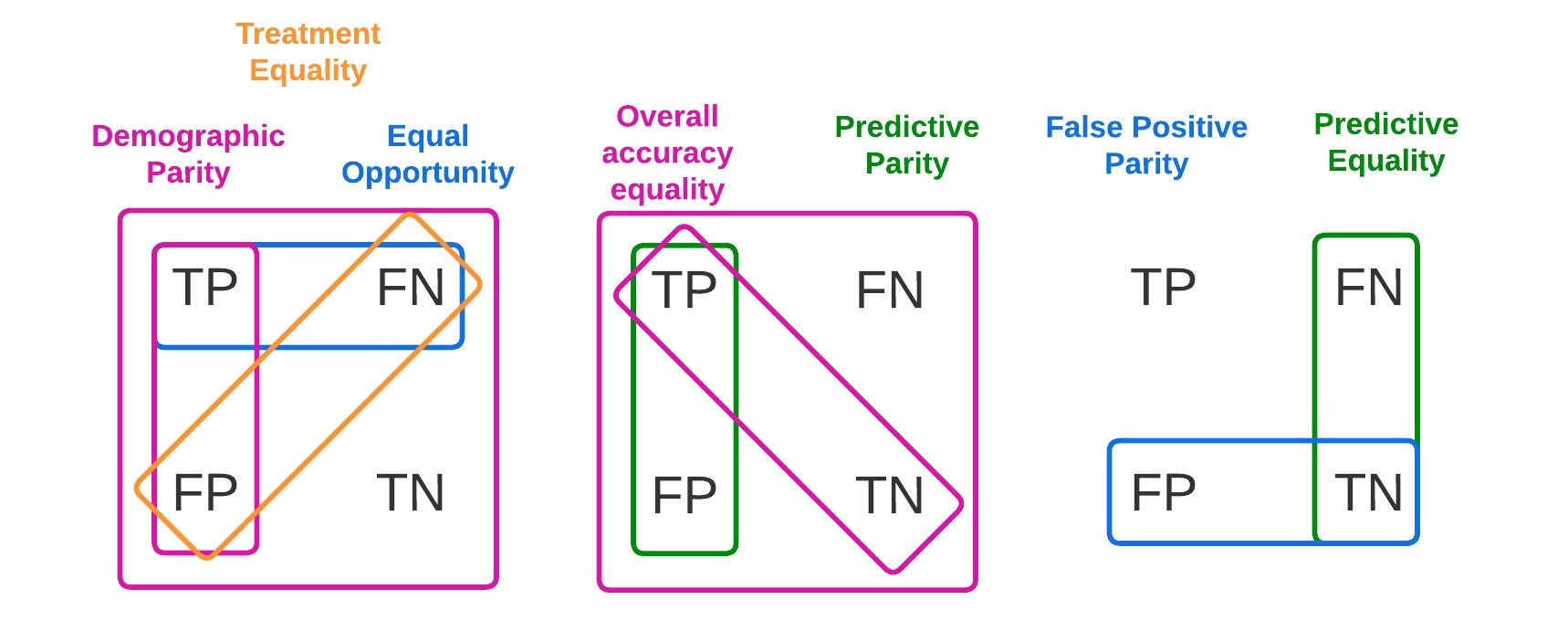}
\caption{Orientations of the fairness measures in the confusion matrix.}
\Description{A visual representations of what areas of the contingency table are effected by the constraint of the fairness measures.}
\label{fig:orient_defs}
\end{figure}

\subsection{Combining two Fairness measures}

We begin with a result on the \emph{pairwise} combinations of the fairness measures discussed in Sec.~\ref{sec:fairnessdefs}.

\begin{figure}[!hb]
\[ 
\begin{cases}
	TP_a + FN_a = 1 - x \\
	TN_a + FP_a = x \\
	TP_b + FN_b = 1 - y \\
	TN_b + FP_b = y \\
	\frac{TP_a}{TP_a + FN_a} = \frac{TP_b}{TP_b + FN_b} \\ 
	\frac{TP_a + FP_a}{TP_a + FP_a + FN_a + TN_a} = \frac{TP_b + FP_b}{TP_b + FP_b + FN_b + TN_b}
\end{cases} 
\Longleftrightarrow
\begin{cases}
	TP_a + FN_a = 1 - x \\
	TN_a + FP_a = x \\
	TP_b + FN_b = 1 - y \\
	TN_b + FP_b = y \\
	\frac{TP_a}{1-x} = \frac{TP_b}{1-y} \\ 
	TP_a + FP_a = TP_b + FP_b
\end{cases}
\]

\[ 
\Longleftrightarrow
\begin{cases}
	TP_a + FN_a = 1 - x \\
	TN_a + FP_a = x \\
	TP_b + FN_b = 1 - y \\
	TN_b + FP_b = y \\
	TP_a = \frac{1-x}{1-y} TP_b \\ 
	FP_a = TP_b + FP_b - \frac{1-x}{1-y} TP_b
\end{cases} 
\Longleftrightarrow
\begin{cases}
	TP_a + FN_a = 1 - x \\
	TN_a + FP_a = x \\
	TP_b + FN_b = 1 - y \\
	TN_b + FP_b = y \\
	TP_a = \frac{1-x}{1-y} TP_b \\ 
	FP_a = (1-\frac{1-x}{1-y}) TP_b + FP_b
\end{cases} 
\Longleftrightarrow
\begin{cases}
	FN_a = 1-x-\frac{1-x}{1-y} TP_b \\ 
	TN_a = x - \frac{x-y}{1-y}TP_b - FP_b \\
	FN_b =  1 - y -TP_b \\ 
	TN_b = y - FP_b \\
	TP_a = \frac{1-x}{1-y} TP_b \\ 
	FP_a = \frac{x-y}{1-y}TP_b + FP_b
\end{cases}
\]
\caption{Combining the fairness measures of Demographic Parity with Equal Opportunity.}
\Description{A system of equations where the constraint of the base rates on the variables are denoted, but also the range in which each variable needs to lie. }
\label{eq:combine_dem_par_eq_opp}
\end{figure}

\begin{figure}
\[
\begin{cases}
	0 \le 1-x-\frac{1-x}{1-y} TP_b \\ 
	1-x-\frac{1-x}{1-y} TP_b \le 1 - x \\	
	0 \le x - \frac{x-y}{1-y}TP_b - FP_b \\
	x - \frac{x-y}{1-y}TP_b - FP_b \le x \\
 
	0 \le \frac{x-y}{1-y}TP_b + FP_b	\\
	\frac{x-y}{1-y}TP_b + FP_b \le x 
\end{cases}
\Longleftrightarrow
\begin{cases}
	\frac{1}{1-y} TP_b \le 1 \\ 
	\frac{1-x}{1-y} TP_b \ge 0 \\	
	FP_b \le x + \frac{-x+y}{1-y}TP_b \\
	\frac{-x+y}{1-y}TP_b  \le FP_b \\
	\frac{-x+y}{1-y}TP_b \le  FP_b	\\
	 FP_b \le x + \frac{-x+y}{1-y}TP_b
\end{cases}
\] 
\caption{Shortened version of calculating the constraints on x and y when combining Demographic Parity with Equal Opportunity, full system of inequalities can be found in Appendix~\ref{subsub:appl_inequalities}. }
\Description{Calculating the constraints on x and y when combining Demographic Parity with Equal Opportunity. Most of the constraints are just properties of the variables and thus insignificant. Four of the constraints however do constrain the value of x, given y and the value of the free variables}
\label{eq:constraint_dem_par_eq_opp}
\end{figure}

\begin{figure}[!h]
\includegraphics[width=0.85\linewidth]{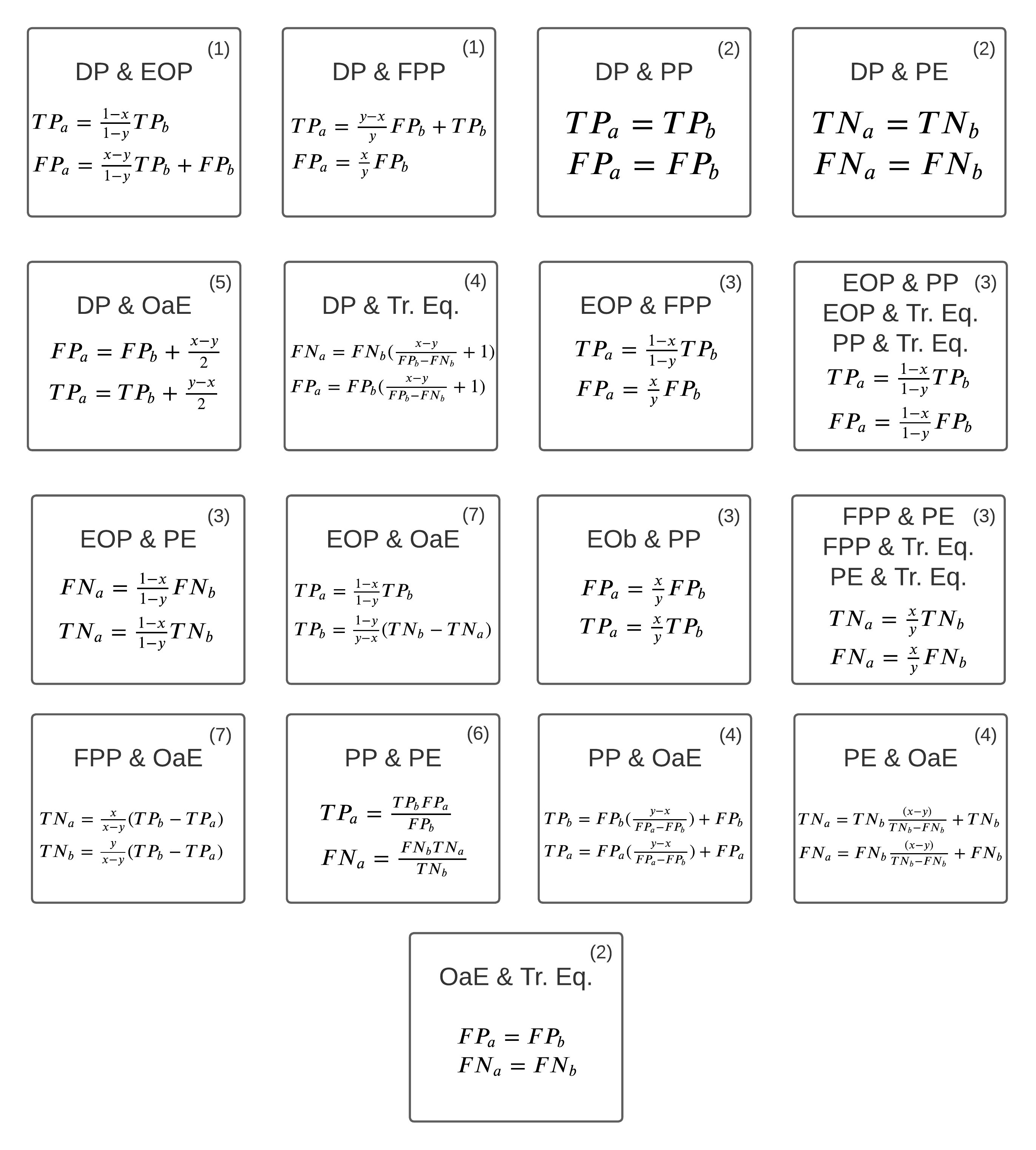}
\caption{Resulting constraints when combining pairs of fairness measures.}
\Description{17 square boxes in which the resulting constraints are denoted in order to satisfy the two fairness measures denoted on the top of the box. In the top right corner of each box a number is added which denotes the kind of group in which a combined measure lies. }
\label{fig:combine_2_def}
\end{figure}
\begin{proposition} \label{prop:comb_2_def}
All pairwise combinations of demographic parity, equal opportunity, false positive parity, predictive parity, predictive equality, overall accuracy equality and treatment equality are possible according to Definition~\ref{def:possible}.
\end{proposition}

\begin{proof_outl}
To prove this for a given pair of fairness measures, we start with the constraints that originate from the confusion matrices from Fig.~\ref{eq:constraints}, and add the respective fairness constraints to yield a combined system of equations. We manipulate this system of equations to express a subset of the variables in terms of the other variables, which remain free. To ensure that the base rates are unconstrained (as required for Definition~\ref{def:possible}), they must be among these other free variables.

An example of calculating the constraints for combining two fairness measures can be found in Figure~\ref{eq:combine_dem_par_eq_opp}, for the case of combining demographic parity with equal opportunity. Using the inequalities from Figure~\ref{eq:constraints} a range can be found for the free variables and the base rates, as is done in Figure~\ref{eq:constraint_dem_par_eq_opp}. This system of inequalities shows a specific range for $FP_b$, the other free variable $TP_b$, and both base rates $x$ and $y$ are not constrained further. The resulting range for $FP_b$ is $[max(0, \frac{-x+y}{1-y}), min(x+\frac{-x+y}{1-y}, y)]$ which simplifies to $[0,x+\frac{-x+y}{1-y}]$ if $x \ge y$, resulting in a non-zero Lebesgue measure, or $[\frac{-x+y}{1-y}, y]$ if $y>x$ then $TP_b$ will have a non-zero Lebesgue measure if $x > y^2$. Thus the Lebesgue measure is non-zero for each free variable for a set of base rates with a non-zero 2-dimensional Lebesgue measure, meaning that the combination of demographic parity and equal opportunity is possible.


The proofs for the other pairwise combinations of fairness measures follow a similar structure. The calculations are in Appendix~\ref{sec:2_def}.\hfill \ensuremath{\Box}
\end{proof_outl}

Note that Proposition~\ref{prop:comb_2_def} indicates that none of the fairness measures are contradictory to each other. It is perhaps unsurprising that what is fair according to one definition is not considered impossibly fair according to another. In fact, two sets of three pairwise combinations are even equivalent with each other, such that they are essentially already combinations of three fairness measures: EOP, PP, and Tr. Eq. form one such group, and FPP, PE, and Tr. Eq. forms the other one.

Two constraints in each system of equations are derived from the constraints of the fairness measures and are unique to them. 
These two resulting constraints for each possible pairwise combination can be found in Figure~\ref{fig:combine_2_def}. It is clear that certain combinations result in constraints that are more easily interpreted. An exception was made for predictive parity and predictive equality as the result becomes fairly convoluted. The constraint can be found in the calculations of combining predictive parity, predictive equality and overall accuracy equality in Section~\ref{subsec:pp_pe_oae}.

Some structure can be found between these two resulting constraints. This structure leads to defining seven types of constraint combinations. These seven types are denoted in the top right corner of the combination in Figure~\ref{fig:combine_2_def}. They serve as a kind of check for the calculations, where the same orientations of the combined fairness definitions would lead to the same type. A summary of the structures and associated measure orientations can be found in Table~\ref{tab:2_dim_types}. The elements in the \textit{Factor} term are denoted by its subscript.
\begin{table}
\begin{center}
\caption{Table showing the structure of the seven different constraint types identified and the combination of orientations associated with the type.}
\Description{This table contains the seven different constraint types and the associated structure of these constraints. The orientations of the fairness measures who's combination of constraints fall into the type are also added.}
\label{tab:2_dim_types}
\begin{tabular}{|l|l|l|}
\hline
Nr. & Orientation of the combined measures & Constraint structure \\ \hline \hline
(1) & $B_vH$ & $X_a = Factor_{BR}*X_b$  \\
 & & $Y_b = Factor_{BR}*X_b + Y_b$ \\ \hline
(2) & $B_vV$, $B_dD$ & $X_a = X_b$ \\
 & & $Y_a = Y_b$ \\ \hline
(3) & $HH, HV, HD, VD$ & $X_a = Factor_{BR}*X_a$  \\
 & & $Y_a = Factor_{BR}*Y_b$ \\ \hline
(4) & $B_vD$, $VB_d$ & $X_a = Factor_{BR, X_b, Y_b} * X_b $ \\
& & $Y_a = Factor_{BR, X_b, Y_b} * Y_b $ \\ \hline 
(5) & $B_vB_d$ & $ X_a = X_b + Factor_{BR} $ \\
& & $Y_a = Y_b + Factor_{BR}$ \\ \hline
(6) & $VV$ & $X_a = Factor_{BR, Y_a, Y_b}*Y_b + Factor_{BR, Y_a, Y_b}*Y_a $ \\
& & $X_b = Factor_{BR, Y_a, Y_b}*Y_b + Factor_{BR, Y_a, Y_b}*Y_a $ \\ \hline
(7) & $HB_d$& $X_a = Factor_{BR}*(Y_b - Y_a)$ \\
& & $X_b = Factor_{BR}*(Y_b - Y_a)$ \\ 
\hline

\end{tabular}
\end{center}
\end{table}

\subsection{Combining three Fairness measures}

We now proceed to characterize the set of possible combinations of three fairness measures.

\begin{proposition} \label{prop:3_comb}
	For combining three measures out of demographic parity, equal opportunity, false positive parity, predictive parity, predictive equality, overall accuracy equality,  and treatment equality only five out of 35 different combinations are possible (according to Definition~\ref{def:possible}).
\end{proposition}
\begin{figure}
\input{Proofs_combine_3/Dem_par_eq_odds_shortened}
\caption{Combining demographic parity, equal opportunity and false positive parity resulting in a solution that has a free variable.}
\Description{The calculations made in order for combining demographic parity, equal opportunity and false positive parity, this can also be found in a formula in the appendix }
\label{fig:dem_par_eq_op_eq_odds_b}
\end{figure}

\begin{figure}
\input{Proofs_combine_3/Dem_par_Pr_par_Pr_eq}
\caption{Combining demographic parity, predictive parity and predictive equality resulting in a solution that requires equal base rates}
\Description{The calculations made in order for combining demographic parity, predictive parity and predict equality, this can also be found in calculations in the appendix}
\label{fig:dem_par_pr_par_pr_eq}
\end{figure}

\begin{figure}
\input{Proofs_combine_3/Dem_par_Eq_opp_Pr_par_shortened}
\caption{Combining demographic parity, equal opportunity and predictive parity}
\Description{A shortened version of the calculations made in order for combining demographic parity, equal opportunity and predictive parity, the full length calculations can be found in the appendix. }
\label{fig:dem_par_eq_opp_pr_par}
\end{figure}

\begin{proof_outl}
In Sec. 3.2 we saw that some pairwise combinations are equivalent with each other, in this way already giving rise to two combinations with three fairness measures. Besides these, combining three fairness measures can be done most conveniently by combining the constraints of two pairwise combinations with one fairness measure in common. As an example, Figure~\ref{fig:dem_par_eq_op_eq_odds_b} shows the calculations for combining demographic parity, equal opportunity, and false positive parity. The resulting constraints show that one variable, $TP_b$, is free, the base rates are unconstrained, and none of the variables are fixed to a particular value. This shows that the combination of demographic parity, equal opportunity and false positive parity is possible.

However, not all combinations of three fairness measures are possible. An example is the combination of demographic parity, predictive parity, and predictive equality. The calculations can be found in Figure~\ref{fig:dem_par_pr_par_pr_eq}. The final set of constraints requires $x=y$, which violates the condition of the base rates having a non-zero 2-dimensional Lebesgue measure. It means that these three fairness measures can only be achieved simultaneously when the base rates happen to be equal to each other, which is entirely data-dependent and not within the control of the algorithm. This means that the combination of demographic parity, predictive parity and predictive equality is not possible.

Another example of how a combination can be impossible is demonstrated in the combination of demographic parity, equal opportunity and predictive parity. The calculations thereof are shown in Figure~\ref{fig:dem_par_eq_opp_pr_par}. The solution either requires equal base rates, which violates the requirements to be deemed possible, or it requires for $TP_b$ and $TP_a$ to equal zero. The result is that four out of eight variables in the confusion matrices have a Lesbesgue measure of zero, violating the requirements for being a possible combination.

The final possible kind of outcome, when combining three fairness measures from the ones used in this paper, arises when combining predictive parity, predictive equality and overall accuracy equality. In this case, the following second degree polynomial needs to be satisfied:
\begin{equation} \label{eq:second_degree}
(-2x+1)TN_b^2 + (xy - y + x + 2yTN_a-x^2 +2xTN_a-2TN_a)TN_b + (xy - x + TN_a -2yTN_a + y - y^2)TN_a = 0.
\end{equation}
As this equation is non-transparent, we used a numerical approach to solve it. We then investigated the results in plots as in Figure~\ref{fig:2nd_degree}. Using this approach it is clear that a free variable, e.g. $TN_a$, exists. It is solvable for a range of base rates, although interestingly not for \emph{all} base rates. However, this range of base rates consists of non-trivial real intervals, such that it satisfies Definition~\ref{def:possible}. On the plots it is evident that none of the variables have a Lebesgue measure of zero, such that the combination of predictive parity, predictive equality and overall accuracy equality is possible.

The calculations for all other combinations of three fairness measures can be found in Appendix~\ref{sec:app_3_def}, and the results are summarized in Table~\ref{tab:char_sol_combine_3}. \hfill \ensuremath{\Box}
\end{proof_outl}

\begin{figure}
\centering
\begin{subfigure}{.5\textwidth}
  \centering
  \includegraphics[width=.9\linewidth]{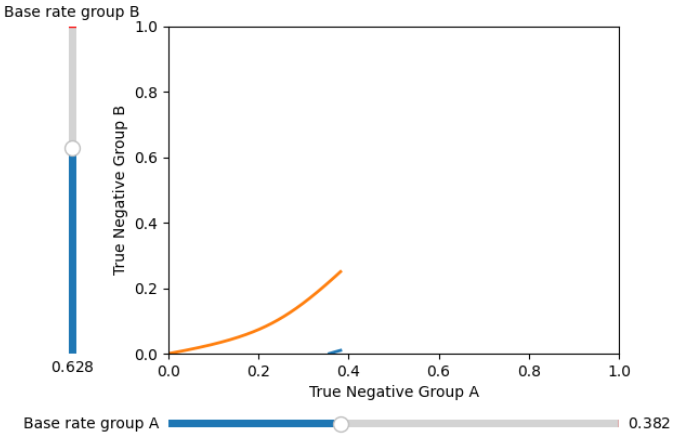}
  \caption{Relation true negative a versus true negative b}
  \label{fig:sub1}
\end{subfigure}%
\begin{subfigure}{.5\textwidth}
  \centering
  \includegraphics[width=.8\linewidth]{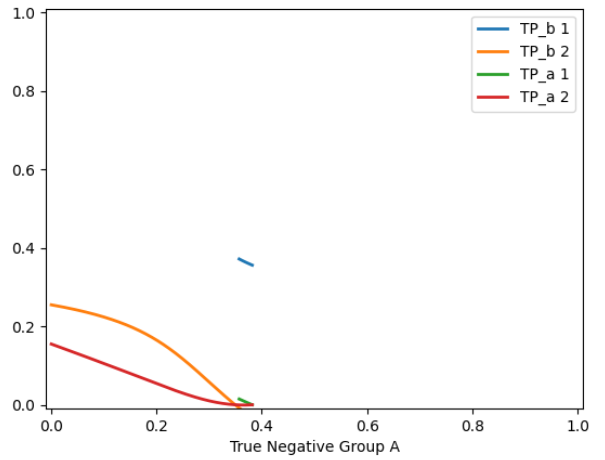}
  \caption{Relation true negative a versus true positive of a and b}
  \label{fig:sub2}
\end{subfigure}
\caption{Relation between the free variable $TN_a$ and $TN_b, TP_a, TP_b$ when combining predictive parity with predictive equality and overall accuracy equality.}
\Description{Two figures showing the program used to calculate the results when combining predictive parity with predictive equality and overall accuracy equality. The left figure shows the two base rate sliding bars designed for input and the graph shows the true negative value for group A related to the true negative group B values. It shows that at a certain point two lines can be drawn relating the two. The right-hand figure shows true negative of group a plotted against the true positive value of b and the true positive value of a. This graph again shows the double possibility of them being related.}
\label{fig:2nd_degree}
\end{figure}

\begin{table}
\begin{center}
\caption{Table containing whether satisfying a certain set of three fairness measures is possible or impossible.}
\label{tab:char_sol_combine_3}
\Description{This table separates all the possible triplet combinations of fairness definitions into a possible or an impossible combinations. The first column denotes possible fairness measures combinations with the final three containing impossible combinations of fairness measures. }
\begin{tabular}{|l|p{2.5cm}|p{2.5cm}|p{2.5cm}|p{2.5cm}|}
\hline.
\textbf{Possible combinations} &  \multicolumn{4}{c|}{\textbf{Impossible combinations}} \\ \hline
DP $|$ EOP $|$ FPP & DP $|$ PP $|$ PE & DP $|$ EOP $|$ PP & DP $|$ PE $|$ Tr. Eq. &  FPP $|$ PP $|$ PE\\ \hline
EOP $|$ FPP $|$ OaE & DP $|$ PP $|$ OaE & DP $|$ EOP $|$ PE &  DP $|$ OaE $|$ Tr. Eq. & FPP $|$ PP $|$ OaE \\ \hline
EOP $|$ PP $|$ Tr. Eq. & DP $|$ PE $|$ OaE & DP $|$ EOP $|$ OaE  & EOP $|$ FPP $|$ PP & FPP $|$ OaE $|$ Tr. Eq. \\ \hline
FPP $|$ PE $|$ Tr. Eq. & DP $|$ PP $|$ Tr. Eq. & DP $|$ FPP $|$ Tr. Eq. & EOP $|$ FPP $|$ Tr. Eq. & FPP $|$ OaE $|$ PE \\ \hline
PP $|$ PE $|$ OaE & EOP $|$ PE $|$ Tr. Eq. & DP $|$ FPP $|$ PP &  EOP $|$ FPP $|$ PE  & FPP $|$ PP $|$ Tr. Eq. \\ \hline
 & EOP $|$ OaE $|$ PP & DP $|$ FPP $|$ PE &  EOP $|$ PE $|$ OaE & PP $|$ PE $|$ Tr. Eq. \\ \hline
 & EOP $|$ OaE $|$ Tr. Eq. & DP $|$ FPP $|$ PE & EOP $|$ PE $|$ PP & PP $|$ OaE $|$ Tr. Eq. \\ \hline
  & PE $|$ OaE $|$ Tr. Eq. & DP $|$ FPP $|$ Tr. Eq. &   &  \\ \hline

\end{tabular}
\end{center}
\end{table}

\subsection{Combining four fairness measures}

As there are only four free variables in the two confusion matrices, due to inherent constraints as in Fig.~\ref{eq:constraints}, combinations of four \emph{independent} fairness notions will not be possible. In other words: for a combination of four fairness measures to be possible, at least one of the constraints should mathematically follow from the combination of other constraints in order for the values in the confusion matrix to have a non-zero Lebesgue measure. This is not the case for the fairness measures considered in this paper.

\begin{proposition}\label{prop:4_comb} 
Out of the fairness measures demographic parity, equal opportunity, false positive parity, predictive parity, predictive equality, overall accuracy equality and treatment equality no sets of four fairness measures are possible (according to Definition~\ref{def:possible}).
\end{proposition}
\begin{proof}
If a set of four fairness measures were possible, then all subsets of size three of these four measures would be possible. However, based on Table~\ref{tab:char_sol_combine_3}, it is easy to exhaustively verify that this is not the case: no set of size four exists for which all its size three subsets are possible. Thus, no set of four fairness measures is possible.
\end{proof}

\subsection{Integration of the results}

Combining the propositions above, we can now state the main result of this paper.

\begin{theorem}
	There are twelve sets that form a maximal possible combinations of fairness measures that are possible according to Definition~\ref{def:possible}, out of demographic parity, equal opportunity, false positive parity, predictive parity, predictive equality, overall accuracy equality and treatment equality. We call these twelve combinations the \emph{maximal fairness} notions. Seven of these sets are pairwise combinations and five are a combination of three fairness measures. Of course, all subsets from these twelve sets are also possible.
\end{theorem}
\begin{proof}
From Propositions~\ref{prop:comb_2_def}, \ref{prop:3_comb}, and \ref{prop:3_comb} Figure~\ref{fig:combination_all_def} is constructed, where each element is a possible combination of fairness measures. It also shows the subset relationships between the possible sets of fairness measures. The 12 red coloured elements indicate that they are maximal sets with regard to possible fairness measure combinations.
\end{proof}

\begin{figure}
\includegraphics[angle=90, origin=c, height=15.05cm]{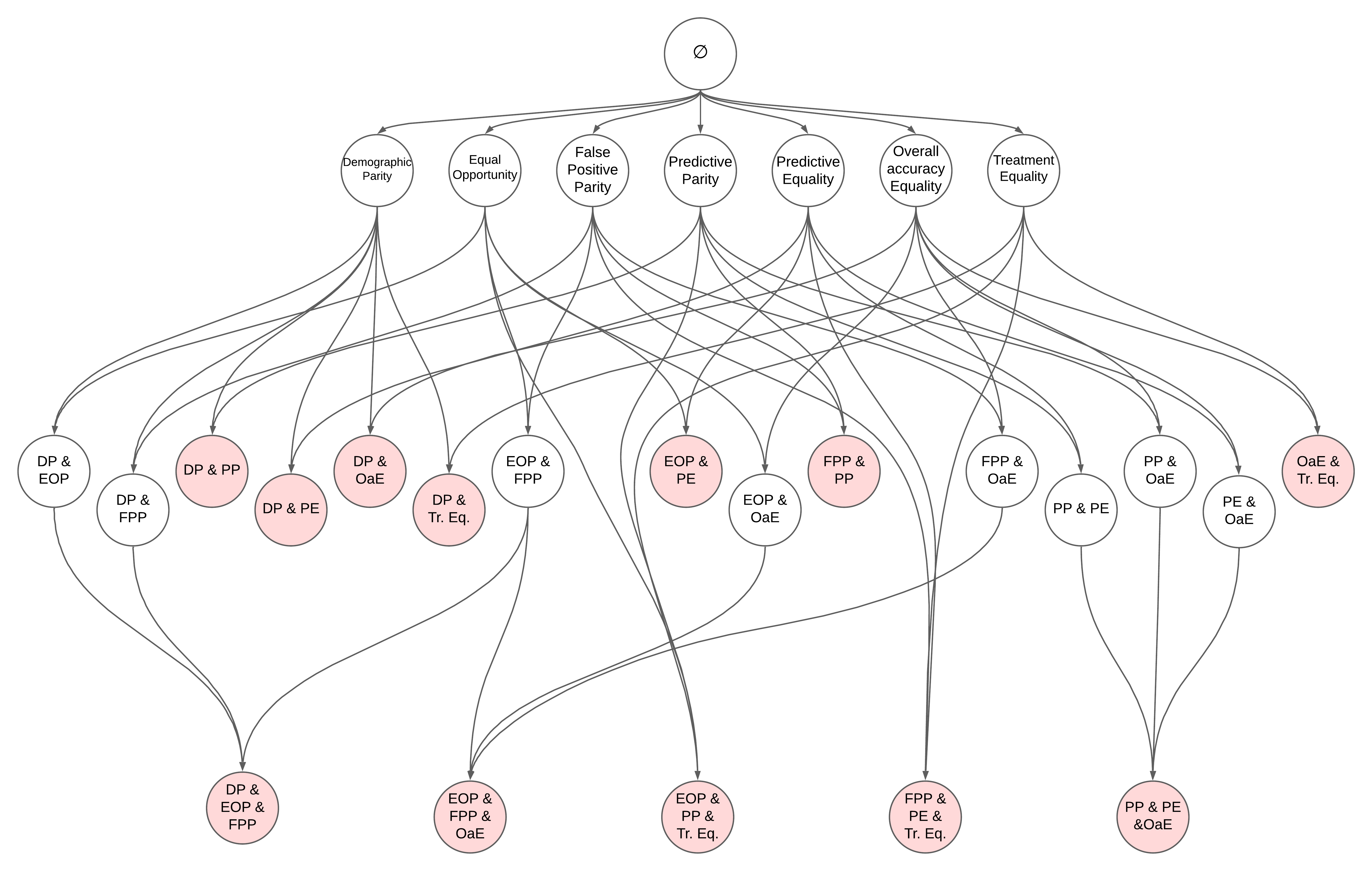}
\caption{All possible combinations of fairness measures.}
\Description{On this figure all possible combinations of fairness measures are combined and represented in a bubble. This means that when combining these measures there is still some freedom for the predictor to work within. The bubbles are connected where a smaller set of measures is connected to a larger set of measures if that smaller set is a subset of the larger set. Some of the bubbles are coloured red denoted that they are a leaf node, meaning that the set cannot be expanded with a new measure without violating the definition of possible.}
\label{fig:combination_all_def}
\end{figure}

Next we try to get some insight into the combinations that are possible.
To this end, Table~\ref{tab:preval_def} shows the number of times each fairness measure can be combined with two others, and the number of times it is impossible because it requires a trivial 2-dimensional range for base rates, or because it requires a trivial range for variables or a trivial 2-dimensional range for base rates. The table reveals an interesting symmetry between equal opportunity and false positive parity and also between predictive parity and predictive equality. This symmetry is due to their mathematical symmetry, as formalized by the following proposition.

\begin{proposition} \label{prop:inversions}
Given any possible combination of fairness measures, another possible combination can be obtained by swapping the semantics of the classes, i.e. by renaming the positive class into the negative and vice versa.
\end{proposition}
Note that intuitively, the swapping of the positive and negative classes can be visualized as a point inversion of the contingency tables through its center.

\begin{proof}
In Fig.~\ref{fig:mirror_def} the point inversions of all possible combinations of three fairness measures are illustrated, showing the symmetry relations between them.


A formal proof of correctness follows from the symmetry relations between the \emph{individual} fairness measures. If the positives are switched with negatives in the constraint of equal opportunity, it becomes the constraint for false positive parity, and vice versa. The same relation holds between predictive parity and predictive equality. The demographic parity constraint, however, changes into a constraint that is equivalent to the original one, and the measures of overall accuracy equality and treatment equality remain unaffected, by switching the positive and negative classes. 
\end{proof}

Proposition~\ref{prop:inversions} also provides an explanation for the symmetry between equal opportunity and false positive parity and between predictive parity and predictive equality. If for every possible combination their inversion is also a possible combination, then it means that for every possible combination with predictive parity there will be a possible combination with predictive equality. Something analogous holds for equal opportunity and false positive parity. This explains why in Table~\ref{tab:preval_def} they have the same number of occurrences in each of the categories.

Note that Fig.~\ref{fig:mirror_def} shows that among the five maximal combinations of three fairness measures, three treat the positive and negative classes equally: they remain unaltered by swapping the role of the positive and negative class. These are (DP, EOP, FPP), (EOP, PP, AoE), and (PP, PE, AoE). The other two possible combinations of three fairness measures do treat the classes differently, and they represents each other's symmetrical analogue. These are (EOP, PP, Tr. Eq.) and (FPP, PE, Tr. Eq.).

\begin{table}
\begin{center}
\caption{Table containing the prevalence of each fairness measure for a characteristic when combined with two others.}
\Description{This table contains the prevalence of each fairness measure in different categories when combining three fairness measures. The first category requires that the constraints still allow for one variable to be completely free. The second category requires that the base rates are equal otherwise it is not possible to satisfy all definitions. The third category has constraints where at least one of the variables has only occurs in a trivial range or the base rates are in a trivial 2-dimensional range. Demographic parity has the least occurrences in the first group with one. Equal opportunity has the exact same distribution as false positive parity and together are highly prevalent in the first class. Predictive parity and predictive equality also have the same distribution as each other. Overall accuracy equality is the most prevalent in the first class out of all fairness measures.}
\label{tab:preval_def}
\begin{tabular}{|l|p{2.8cm}|p{2.7cm}|p{3.9cm}|}
\hline
\textbf{Measure} & \textbf{Nr. of times in possible combinations} & \textbf{Nr. of times requiring equal base rates} & \textbf{Nr. of times trivial range for variables or trivial 2-dim. range for base rates} \\ \hline
Demographic Parity & 1 & 3 & 11 \\ \hline
Equal Opportunity & 3 & 2 & 10 \\ \hline
False positive parity & 3 & 2 & 10 \\ \hline
Predictive Parity & 2 & 6 & 7 \\ \hline
Predictive Equality & 2 & 6 & 7 \\ \hline
Overall accuracy Equality & 2 & 2 & 11 \\\hline
Treatment Equality & 2 & 3 & 10 \\\hline

\end{tabular}
\end{center}
\end{table}

\begin{figure}
\includegraphics[origin=c, width=0.75\textwidth]{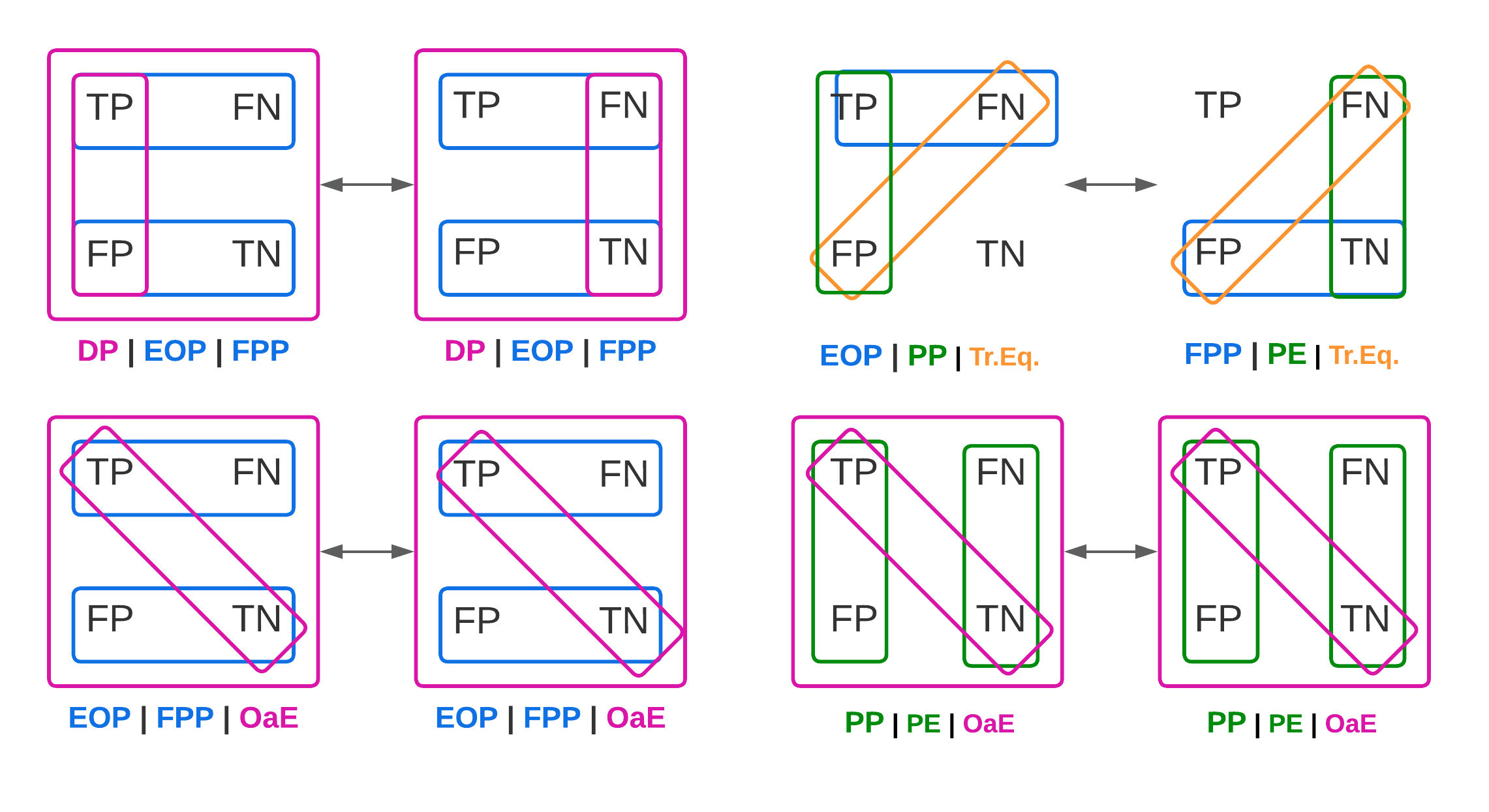}
\caption{Point inversion of the possible combinations of three fairness measures.}
\Description{This figure shows the combinations of three measures that still leave room for some freedom by drawing their orientations and effects on the confusion matrix. Related to that a point of version is done in the middle of the square signifying the confusion matrix and the resulting orientations are drawn from this. This results in 8 squares as three out of 5 combinations when applying the point inversion remain the same combination.}
\label{fig:mirror_def}
\end{figure}

Finally, note that both Figure~\ref{fig:combination_all_def} and Table~\ref{tab:preval_def} indicate that possible combinations with demographic parity are most rare. Demographic parity is a fairness definition that does not relate to the base rates of the groups due to it being independent from the ground truth. This is most likely the reason why it can only be combined with equal opportunity and false positive parity as these definitions are directly related to the base rates.

\section{Conclusion}

We investigated which and how many fairness notions can be combined and imposed simultaneously, for the simple setting of binary classification with two demographic groups. We find that out of seven fairness measures, namely demographic parity, equal opportunity, false positive parity, predictive parity, predictive equality, overall accuracy equality and treatment equality, in total twelve maximal combinations are possible. Five of these maximal combinations consist of three fairness measures, while seven of them consist of just two fairness measures. Each of these sets could be regarded as \emph{maximal fairness} notions, in that it is not possible to impose any further fairness constraints without making the problem infeasible. Of course, all subsets of these twelve maximal sets are also possible combinations of fairness definitions. An overview of the possible combinations, maximal as well as non-maximal, and the subset relations between them, is given in Figure~\ref{fig:combination_all_def}.
Our results confirm and extend related work, such as the theorem colloquially referred to as the fairness impossibility theorem. 

We consider the five maximal sets of three fairness measures as particularly interesting. Of these five, three treat the positive and negative classes symmetrically, in that they remain unaltered after swapping the class labels. The remaining two are each other's symmetrical analogue: one of these two sets of fairness measures is equivalent to the other set after swapping the class labels. This property might provide some initial guidance and insight to determine which maximal fairness notion is most appropriate in a given application: if the costs of different kinds of misclassifications are very different, one of the maximal fairness notions that are not invariant with respect to a label swap might be most appropriate.

An important next step building on these results would be to develop methods that are capable of imposing these (maximal) sets of fairness measures.
Subsequently, it would be interesting to investigate the impact of imposing (maximal) sets of fairness measures on the accuracy and learning capabilities (e.g. learning rate and convergence rate) of machine learning methods.

A separate line of further work would be towards including other fairness measures, in addition to the seven measures included in the present paper. Moreover, we consider extensions towards other problems besides binary classification and just two demographic groups as fruitful avenues for further research.

Finally, and perhaps most importantly, the meaning of the maximal fairness notions derived in the present paper needs to be better understood from a practical, ethical and/or legal point of view. This is an important aspect of research in fairness in AI as this translation to real world concepts is a pre-requisite for wide-spread adoption. 


\bibliographystyle{ACM-Reference-Format}
\bibliography{biblio}

\appendix

 \section{Combining two definitions}
 \label{sec:2_def}

\subsection{Equal Opportunity and Demographic Parity}
$$ 
\begin{cases}
	TP_a + FN_a = 1 - x \\
	TN_a + FP_a = x \\
	TP_b + FN_b = 1 - y \\
	TN_b + FP_b = y \\
	\frac{TP_a}{TP_a + FN_a} = \frac{TP_b}{TP_b + FN_b} \\ 
	\frac{TP_a + FP_a}{TP_a + FP_a + FN_a + TN_a} = \frac{TP_b + FP_b}{TP_b + FP_b + FN_b + TN_b}
\end{cases} 
\Longleftrightarrow
\begin{cases}
	TP_a + FN_a = 1 - x \\
	TN_a + FP_a = x \\
	TP_b + FN_b = 1 - y \\
	TN_b + FP_b = y \\
	\frac{TP_a}{1-x} = \frac{TP_b}{1-y} \\ 
	TP_a + FP_a = TP_b + FP_b
\end{cases}$$ $$ 
\Longleftrightarrow
\begin{cases}
	TP_a + FN_a = 1 - x \\
	TN_a + FP_a = x \\
	TP_b + FN_b = 1 - y \\
	TN_b + FP_b = y \\
	TP_a = \frac{1-x}{1-y} TP_b \\ 
	FP_a = TP_b + FP_b - \frac{1-x}{1-y} TP_b
\end{cases} 
\Longleftrightarrow
\begin{cases}
	TP_a + FN_a = 1 - x \\
	TN_a + FP_a = x \\
	TP_b + FN_b = 1 - y \\
	TN_b + FP_b = y \\
	TP_a = \frac{1-x}{1-y} TP_b \\ 
	FP_a = (1-\frac{1-x}{1-y}) TP_b + FP_b
\end{cases} 
\Longleftrightarrow
\begin{cases}
	FN_a = 1-x- \frac{1-x}{1-y} TP_b \\ 
	TN_a = x - \frac{x-y}{1-y}TP_b - FP_b \\
	FN_b =  1 - y -TP_b \\ 
	TN_b = y - FP_b \\
	TP_a = \frac{1-x}{1-y} TP_b \\ 
	FP_a = \frac{x-y}{1-y}TP_b + FP_b
\end{cases}
$$ 
\subsubsection{Applying inequalities}
\label{subsub:appl_inequalities}
\[
\begin{cases}
	0 \le 1-x-\frac{1-x}{1-y} TP_b \\ 
	1-x-\frac{1-x}{1-y} TP_b \le 1 - x \\	
	0 \le x - \frac{x-y}{1-y}TP_b - FP_b \\
	x - \frac{x-y}{1-y}TP_b - FP_b \le x \\
	0 \le 1 - y -TP_b \\ 
	1 - y -TP_b \le 1-y \\ 
	0 \le y - FP_b \\
	y - FP_b \le y \\
	0 \le \frac{1-x}{1-y} TP_b \\ 
	\frac{1-x}{1-y} TP_b \le 1 - x \\ 
	0 \le \frac{x-y}{1-y}TP_b + FP_b	\\
	\frac{x-y}{1-y}TP_b + FP_b \le x 
\end{cases}
\Longleftrightarrow
\begin{cases}
	\frac{1}{1-y} TP_b \le 1 \\ 
	\frac{1-x}{1-y} TP_b \ge 0 \\	
	FP_b \le x + \frac{-x+y}{1-y}TP_b \\
	\frac{-x+y}{1-y}TP_b  \le FP_b \\
	TP_b \le 1 - y \\ 
	TP_b \ge 0 \\ 
	FP_b \le y \\
	FP_b \ge 0 \\
	0 \le \frac{1-x}{1-y} \\ 
	\frac{1}{1-y} TP_b \le 1 \\ 
	\frac{y-x}{1-y} \le  \frac{FP_b}{TP_b}	\\
	 FP_b + \frac{x-y}{1-y}TP_b \le x 
\end{cases}
 \Longleftrightarrow
\begin{cases}
	\frac{y-x}{1-y} \le  \frac{FP_b}{TP_b}	\\
	 FP_b + \frac{x-y}{1-y}TP_b \le x 
\end{cases}\]

\subsection{Demographic Parity and False positive parity}
$$ \begin{cases}
	TP_a + FN_a = 1 - x \\
	TN_a + FP_a = x \\
	TP_b + FN_b = 1 - y \\
	TN_b + FP_b = y \\
	TP_a + FP_a = TP_b + FP_b \\
	FP_a = \frac{x}{y} FP_b
\end{cases} \Longleftrightarrow 
\begin{cases}
	TP_a + FN_a = 1 - x \\
	TN_a + FP_a = x \\
	TP_b + FN_b = 1 - y \\
	TN_b + FP_b = y \\
	FP_a = \frac{x}{y} FP_b	\\
	TP_a = TP_b + (1-\frac{x}{y}) FP_b
\end{cases} \Longleftrightarrow
\begin{cases}
	FN_a = 1 - x - \frac{y-x}{y}FP_b - TP_b \\
	TN_a = x  - \frac{x}{y} FP_b\\
	FN_b = 1 - y - TP_b \\
	TN_b = y - FP_b \\
	FP_a = \frac{x}{y} FP_b \\ 
	TP_a = \frac{y-x}{y}FP_b + TP_b
\end{cases}
$$

\subsubsection{Added constraints through inequalities}
\[
\begin{cases}
	\frac{x-y}{y} \le  \frac{TP_b}{FP_b} \\
	1-x \ge \frac{y-x}{y}FP_b + TP_b \\
\end{cases} 
 \]

\subsection{Demographic Parity and Predictive Parity}
$$ \begin{cases}
	TP_a + FN_a = 1 - x \\
	TN_a + FP_a = x \\
	TP_b + FN_b = 1 - y \\
	TN_b + FP_b = y \\
	TP_a*FP_b = TP_b*FP_a \\ 
	TP_a + FP_a = TP_b + FP_b
\end{cases} \Longleftrightarrow
\begin{cases}
	TP_a + FN_a = 1 - x \\
	TN_a + FP_a = x \\
	TP_b + FN_b = 1 - y \\
	TN_b + FP_b = y \\
	TP_a = \frac{TP_b*FP_a}{FP_b} \\ 
	\frac{TP_b*FP_a}{FP_b} + FP_a = TP_b + FP_b
\end{cases}
 \Longleftrightarrow
\begin{cases}
	TP_a + FN_a = 1 - x \\
	TN_a + FP_a = x \\
	TP_b + FN_b = 1 - y \\
	TN_b + FP_b = y \\
	TP_a = \frac{TP_b*FP_a}{FP_b} \\ 
	\frac{FP_b + TP_b}{FP_b}*FP_a = TP_b + FP_b
\end{cases} $$

$$ \Longleftrightarrow
\begin{cases}
	TP_a + FN_a = 1 - x \\
	TN_a + FP_a = x \\
	TP_b + FN_b = 1 - y \\
	TN_b + FP_b = y \\
	TP_a = \frac{TP_b*FP_a}{FP_b} \\ 
	FP_a = FP_b \lor FP_b + TP_b = 0
\end{cases}
\Longleftrightarrow
\begin{cases}
	FN_a = 1 - x - TP_b \\
	TN_a = x - FP_b\\
	FN_b = 1 - y - TP_b \\
	TN_b = y - FP_b \\
	TP_a = TP_b\\ 
	FP_a = FP_b
\end{cases} \lor
\begin{cases}
	FN_a = 1 - x - FP_a \\
	TN_a = x - FP_a \\
	FN_b = 1 - y \\
	TN_b = y \\
	TP_a = FP_a \\ 
	FP_b = TP_b = 0
\end{cases}$$

\subsubsection{Added constraints through inequalities}
\[
\begin{cases}
	1 - x \ge TP_b \\
	x \ge FP_b
\end{cases}
\]

\subsection{Demographic parity and predictive equality}
$$ \begin{cases}
	TP_a + FN_a = 1 - x \\
	TN_a + FP_a = x \\
	TP_b + FN_b = 1 - y \\
	TN_b + FP_b = y \\
	FN_a * TN_b = FN_b * TN_a \\
	TP_a + FP_a = TP_b + FP_b 
\end{cases} \Longleftrightarrow
\begin{cases}
	TP_a + FN_a = 1 - x \\
	TN_a + FP_a = x \\
	TP_b + FN_b = 1 - y \\
	TN_b + FP_b = y \\
	FN_a = \frac{FN_b*TN_a}{TN_b} \\
	1 - x - FN_a + x - TN_a = 1 - y - FN_b + y - TN_b
\end{cases}
$$ 
$$ \Longleftrightarrow 
\begin{cases}
	TP_a + FN_a = 1 - x \\
	TN_a + FP_a = x \\
	TP_b + FN_b = 1 - y \\
	TN_b + FP_b = y \\
	FN_a = \frac{FN_b*TN_a}{TN_b} \\
	FN_a + TN_a = FN_b + TN_b
\end{cases} \Longleftrightarrow
\begin{cases}
	TP_a + FN_a = 1 - x \\
	TN_a + FP_a = x \\
	TP_b + FN_b = 1 - y \\
	TN_b + FP_b = y \\
	FN_a = \frac{FN_b*TN_a}{TN_b} \\
	\frac{FN_b*TN_a}{TN_b} + TN_a = FN_b + TN_b
\end{cases}
 \Longleftrightarrow 
\begin{cases}
	TP_a + FN_a = 1 - x \\
	TN_a + FP_a = x \\
	TP_b + FN_b = 1 - y \\
	TN_b + FP_b = y \\
	FN_a = \frac{FN_b*TN_a}{TN_b} \\
	\frac{FN_b + TN_b}{TN_b}*TN_a = FN_b + TN_b
\end{cases} $$
$$\Longleftrightarrow
\begin{cases}
	TP_a + FN_a = 1 - x \\
	TN_a + FP_a = x \\
	TP_b + FN_b = 1 - y \\
	TN_b + FP_b = y \\
	TN_a = TN_b \lor FN_b = TN_b = 0 \\
	FN_a = \frac{FN_b*TN_a}{TN_b}
\end{cases}
\Longleftrightarrow
\begin{cases}
	TP_a = 1 - x - FN_a \\
	FP_a = x - FN_a \\
	TP_b = 1 - y \\
	FP_b = y \\
	FN_b = TN_b = 0 \\
	FN_a = TN_a
\end{cases} \lor
\begin{cases}
	TP_a = 1 - x - FN_b \\
	FP_a = x - TN_b \\
	TP_b = 1 - y - FN_b \\
	FP_b = y - TN_b \\
	TN_a = TN_b \\
	FN_a = FN_b
\end{cases}
$$
\subsubsection{Added constraints through inequalities}
\[
\begin{cases}	
	FN_b \le 1 - x \\
	TN_b \le x
\end{cases} 
\]

\subsection{Demographic Parity and Overall accuracy Equality}
$$ \begin{cases}
	TP_a + FN_a = 1 - x \\
	TN_a + FP_a = x \\
	TP_b + FN_b = 1 - y \\
	TN_b + FP_b = y \\
	TP_a + FP_a = TP_b + FP_b \\
	TP_a + TN_a = TP_b + TN_b
\end{cases} \Longleftrightarrow
\begin{cases}
	TP_a + FN_a = 1 - x \\
	TN_a + FP_a = x \\
	TP_b + FN_b = 1 - y \\
	TN_b + FP_b = y \\
	TP_a - TP_b = FP_b - FP_a \\
	TP_a - TP_b = TN_b - TN_a 
\end{cases}
$$ 

$$ \Longleftrightarrow \begin{cases}
	TP_a + FN_a = 1 - x \\
	TN_a + FP_a = x \\
	TP_b + FN_b = 1 - y \\
	TN_b + FP_b = y \\
	TN_b - TN_a = FP_b - FP_a \\
	TP_a - TP_b = FP_b - FP_a
\end{cases} \Longleftrightarrow
\begin{cases}
	TP_a + FN_a = 1 - x \\
	TN_a + FP_a = x \\
	TP_b + FN_b = 1 - y \\
	TN_b + FP_b = y \\
	y - FP_b - x + FP_a = FP_b - FP_a \\
	TP_a - TP_b = FP_b - FP_a 
\end{cases}
$$ 

$$ \Longleftrightarrow \begin{cases}
	TP_a + FN_a = 1 - x \\
	TN_a + FP_a = x \\
	TP_b + FN_b = 1 - y \\
	TN_b + FP_b = y \\
	y - x = 2FP_b - 2FP_a \\
	TP_a - TP_b = FP_b - FP_a
\end{cases} \Longleftrightarrow
\begin{cases}
	TP_a + FN_a = 1 - x \\
	TN_a + FP_a = x \\
	TP_b + FN_b = 1 - y \\
	TN_b + FP_b = y \\
	FP_b - FP_a = \frac{y-x}{2} \\
	TP_a - TP_b =  \frac{y-x}{2}
\end{cases} \Longleftrightarrow
\begin{cases}
	FN_a = 1 - x - TP_b - \frac{y-x}{2} \\
	TN_a = x - FP_b - \frac{x-y}{2} \\
	FN_b = 1 - y - TP_b \\
	TN_b = y - FP_b \\
	FP_a = FP_b + \frac{x-y}{2} \\
	TP_a = TP_b + \frac{y-x}{2}
\end{cases}
$$ 
\subsubsection{Added constraints through inequalities}
\[
\begin{cases}
	TP_b \le 1 - \frac{x+y}{2} \\
	TP_b \ge \frac{x-y}{2} \\
	FP_b \le \frac{x+y}{2} \\
	FP_b \ge \frac{y-x}{2} \\
\end{cases}
\]

\subsection{Demographic parity and Treatment Equality}
$$ 
\begin{cases}
	TP_a + FN_a = 1 - x \\
	TN_a + FP_a = x \\
	TP_b + FN_b = 1 - y \\
	TN_b + FP_b = y \\
	TP_a + FP_a = TP_b + FP_b \\ 
	\frac{FN_a}{FP_a} = \frac{FN_b}{FP_b}
\end{cases} \Longleftrightarrow
\begin{cases}
	TP_a + FN_a = 1 - x \\
	TN_a + FP_a = x \\
	TP_b + FN_b = 1 - y \\
	TN_b + FP_b = y \\
	FP_a = \frac{FN_a*FP_b}{FN_b} \\
	1 - x - FN_a + \frac{FN_a*FP_b}{FN_b}  = 1 - y - FN_b + FP_b
\end{cases}
$$
$$ \Longleftrightarrow
\begin{cases}
	TP_a + FN_a = 1 - x \\
	TN_a + FP_a = x \\
	TP_b + FN_b = 1 - y \\
	TN_b + FP_b = y \\
	FP_a = \frac{FN_a*FP_b}{FN_b} \\ 
	\frac{FP_b - FN_b}{FN_b} FN_a = x - y + FP_b - FN_b
\end{cases} \Longleftrightarrow
\begin{cases}
	TP_a = 1 - x - FN_b (\frac{x-y}{FP_b - FN_b} + 1) \\
	TN_a = x - FP_b (\frac{x-y}{FP_b - FN_b} + 1) \\
	TP_b = 1 - y - FN_b \\
	TN_b = y - FP_b \\
	FN_a = FN_b (\frac{x-y}{FP_b - FN_b} + 1) \\
	FP_a = FP_b (\frac{x-y}{FP_b - FN_b} + 1)
\end{cases} \lor 
\begin{cases}
	TP_a = 1 - x - FN_a \\
	TN_a = x - FN_a \\
	TP_b = 1 - y - FN_b \\
	TN_b = y - FN_b\\
	FP_b = FN_b \\
	x = y \\
	FP_a = FN_a
\end{cases}
$$
\subsubsection{Added constraints through inequalities}
\[ 
\begin{cases}
	\frac{x-y}{FP_b - FN_b} \ge -1 \\
	1 - x \ge FN_b \\
	x \ge FP_b
\end{cases}
\]

\subsection{Equalised Odds}
$$
\begin{cases}
	TP_a + FN_a = 1 - x \\
	TN_a + FP_a = x \\
	TP_b + FN_b = 1 - y \\
	TN_b + FP_b = y \\
	TP_a = \frac{1-x}{1-y} TP_b \\
	\frac{FP_a}{FP_a + TN_a} = \frac{FP_b}{FP_b + TN_b} 
\end{cases} \Longleftrightarrow
\begin{cases}
	TP_a + FN_a = 1 - x \\
	TN_a + FP_a = x \\
	TP_b + FN_b = 1 - y \\
	TN_b + FP_b = y \\
	TP_a = \frac{1-x}{1-y} TP_b \\
	\frac{FP_a}{x} = \frac{FP_b}{y}
\end{cases}
\Longleftrightarrow
\begin{cases}
	FN_a = 1 - x- \frac{1-x}{1-y} TP_b \\
	TN_a = x - \frac{x}{y} FP_b \\	
	FN_b = 1 - y - TP_b \\
	TN_b = y - FP_b \\	
	TP_a = \frac{1-x}{1-y} TP_b \\ 
	FP_a = \frac{x}{y} FP_b
\end{cases} $$

\subsection{Equal Opportunity and Predictive Parity}
$$ \begin{cases}
	TP_a + FN_a = 1 - x \\
	TN_a + FP_a = x \\
	TP_b + FN_b = 1 - y \\
	TN_b + FP_b = y \\
	TP_a = \frac{1-x}{1-y} TP_b \\ 
	\frac{TP_a}{TP_a + FP_a} = \frac{TP_b}{TP_b + FP_b}
\end{cases} \Longleftrightarrow 
\begin{cases}
	TP_a + FN_a = 1 - x \\
	TN_a + FP_a = x \\
	TP_b + FN_b = 1 - y \\
	TN_b + FP_b = y \\
	TP_a = \frac{1-x}{1-y} TP_b \\ 
	TP_a * FP_b = TP_b * FP_a
\end{cases} \Longleftrightarrow \begin{cases}
	TP_a + FN_a = 1 - x \\
	TN_a + FP_a = x \\
	TP_b + FN_b = 1 - y \\
	TN_b + FP_b = y \\
	TP_a = \frac{1-x}{1-y} TP_b \\ 
	\frac{1-x}{1-y} * TP_b * FP_b = TP_b * FP_a
\end{cases} 
$$
$$\Longleftrightarrow \begin{cases}
	TP_a + FN_a = 1 - x \\
	TN_a + FP_a = x \\
	TP_b + FN_b = 1 - y \\
	TN_b + FP_b = y \\
	TP_a = \frac{1-x}{1-y} TP_b \\ 
	\frac{1-x}{1-y} FP_b = FP_a \lor TP_b = 0
\end{cases} \Longleftrightarrow \begin{cases}
	FN_a = 1 - x - \frac{1-x}{1-y} TP_b \\
	TN_a = x - \frac{1-x}{1-y} FP_b \\
	FN_b = 1 - y - TP_b \\
	TN_b = y - FP_b \\
	TP_a = \frac{1-x}{1-y} TP_b \\ 
	FP_a = \frac{1-x}{1-y} FP_b
\end{cases} \lor 
\begin{cases}
	FN_a = 1 - x \\
	TN_a = x - FP_a \\
	FN_b = 1 - y \\
	TN_b = y - FP_b \\
	TP_a = 0 \\ 
	TP_b = 0
\end{cases}
$$ 
\subsubsection{Added constraint through inequalities}
\[ FP_b \le x\frac{1-y}{1-x}FP_b \]

\subsection{Equal Opportunity and Predictive Equality}
$$ \begin{cases}
	TP_a + FN_a = 1 - x \\
	TN_a + FP_a = x \\
	TP_b + FN_b = 1 - y \\
	TN_b + FP_b = y \\
	TP_a = \frac{1-x}{1-y} TP_b \\ 
	\frac{FN_a}{FN_a + TN_a} = \frac{FN_b}{FN_b + TN_b}
\end{cases} \Longleftrightarrow 
\begin{cases}
	TP_a + FN_a = 1 - x \\
	TN_a + FP_a = x \\
	TP_b + FN_b = 1 - y \\
	TN_b + FP_b = y \\
	TP_a = \frac{1-x}{1-y} TP_b \\ 
	FN_a* TN_b = FN_b*TN_a
\end{cases}
\Longleftrightarrow \begin{cases}
	TP_a + FN_a = 1 - x \\
	TN_a + FP_a = x \\
	TP_b + FN_b = 1 - y \\
	TN_b + FP_b = y \\
	1 - x - FN_a = \frac{1-x}{1-y} (1-y-FN_b) \\ 
	FN_a* TN_b = FN_b*TN_a
\end{cases}$$ 

$$
 \Longleftrightarrow 
\begin{cases}
	TP_a + FN_a = 1 - x \\
	TN_a + FP_a = x \\
	TP_b + FN_b = 1 - y \\
	TN_b + FP_b = y \\
	FN_a = \frac{1-x}{1-y}FN_b \\
	\frac{1-x}{1-y}FN_b * TN_b = FN_b*TN_a
\end{cases}
\Longleftrightarrow \begin{cases}
	TP_a + FN_a = 1 - x \\
	TN_a + FP_a = x \\
	TP_b + FN_b = 1 - y \\
	TN_b + FP_b = y \\
	FN_a = \frac{1-x}{1-y}FN_b \\ 
	\frac{1-x}{1-y} TN_b = TN_a \lor FN_b = 0
\end{cases} $$
$$
\Longleftrightarrow 
\begin{cases}
	TP_a = 1 - x - \frac{1-x}{1-y} FN_b  \\
	FP_a = x - \frac{1-x}{1-y} TN_b  \\
	TP_b = 1 - y - FN_b \\
	FP_b = y - TN_b \\
	FN_a = \frac{1-x}{1-y} FN_b \\
	TN_a = \frac{1-x}{1-y} TN_b 
\end{cases} \lor 
\begin{cases}
	TP_a = 1 - x  \\
	FP_a = x - TN_a \\
	TP_b = 1 - y \\
	FP_b = y - TN_b \\
	FN_a = 0 \\
	FN_b = 0 
\end{cases}
$$ 
\subsubsection{Added constraint through inequalities}
\[
TN_b \le \frac{1-y}{1-x}*x
\]

\subsection{Equal opportunity and Overall accuracy equality}
$$ \begin{cases}
	TP_a + FN_a = 1 - x \\
	TN_a + FP_a = x \\
	TP_b + FN_b = 1 - y \\
	TN_b + FP_b = y \\
	TP_a = \frac{1-x}{1-y} TP_b \\ 
	TP_a + TN_a = TP_b + TN_b 
\end{cases} 
\Longleftrightarrow 
\begin{cases}
	TP_a + FN_a = 1 - x \\
	TN_a + FP_a = x \\
	TP_b + FN_b = 1 - y \\
	TN_b + FP_b = y \\
	TP_a = \frac{1-x}{1-y} TP_b \\ 
	\frac{1-x}{1-y} TP_b - TP_b = TN_b - TN_a 
\end{cases} $$ $$
\Longleftrightarrow 
\begin{cases}
	TP_a + FN_a = 1 - x \\
	TN_a + FP_a = x \\
	TP_b + FN_b = 1 - y \\
	TN_b + FP_b = y \\
	TP_a = \frac{1-x}{1-y} TP_b \\ 
	\frac{y-x}{1-y} TP_b = TN_b - TN_a 
\end{cases} 
\Longleftrightarrow 
\begin{cases}
	FN_a = 1 - x - \frac{1-x}{y-x} (TN_b - TN_a) \\
	FP_a = x - TN_a \\
	FN_b = 1 - y - \frac{1-y}{y-x} (TN_b - TN_a) \\
	FP_b = y - TN_b\\
	TP_a = \frac{1-x}{y-x} (TN_b - TN_a) \\ 
	TP_b = \frac{1-y}{y-x} (TN_b - TN_a)
\end{cases} $$ 

\subsubsection{Added constraints through inequalities}
\[ 
\begin{cases}
	y > x \\
	TN_a \le TN_b \\
	y - x \ge TN_b - TN_a \\
\end{cases} \lor
\begin{cases}
	x > y \\
	TN_b \le TN_b \\
	y - x \le TN_b - TN_a \\
\end{cases}
\]

\subsection{Equal opportunity and Treatment Equality}
$$ \begin{cases}
	TP_a + FN_a = 1 - x \\
	TN_a + FP_a = x \\
	TP_b + FN_b = 1 - y \\
	TN_b + FP_b = y \\
	TP_a = \frac{1-x}{1-y} TP_b \\ 
	\frac{FN_a}{FP_a} = \frac{FN_b}{FP_b} 
\end{cases} 
\Longleftrightarrow 
\begin{cases}
	TP_a + FN_a = 1 - x \\
	TN_a + FP_a = x \\
	TP_b + FN_b = 1 - y \\
	TN_b + FP_b = y \\
	FN_a = \frac{1-x}{1-y} FN_b \\ 
	\frac{1-x}{1-y} \frac{FN_b}{FP_a} = \frac{FN_b}{FP_b}  
\end{cases} 
\Longleftrightarrow 
\begin{cases}
	TP_a + FN_a = 1 - x \\
	TN_a + FP_a = x \\
	TP_b + FN_b = 1 - y \\
	TN_b + FP_b = y \\
	FN_a = \frac{1-x}{1-y} FN_b \\ 
	FP_a = \frac{1-x}{1-y} FP_b \lor FN_b = 0
\end{cases} $$ $$
\Longleftrightarrow 
\begin{cases}
	TP_a = 1 - x - \frac{1-x}{1-y} FN_b \\
	TN_a = x - \frac{1-x}{1-y} FP_b \\
	TP_b = 1 - y - FN_b \\
	TN_b = y - FP_b \\
	FN_a = \frac{1-x}{1-y} FN_b \\ 
	FP_a = \frac{1-x}{1-y} FP_b 
\end{cases} \lor
\begin{cases}
	TP_a = 1 - x \\
	TN_a = x - FP_a \\
	TP_b = 1 - y  \\
	TN_b = y - FP_b \\
	FN_a = 0 \\ 
	FN_b = 0 
\end{cases}
 $$ 
 
\subsubsection{Added constraints through inequalities}
\[
\frac{1-y}{1-x}x \ge FP_b 
\]

\subsection{False positive parity and Predictive Parity}
$$ 
\begin{cases}
	TP_a + FN_a = 1 - x \\
	TN_a + FP_a = x \\
	TP_b + FN_b = 1 - y \\
	TN_b + FP_b = y \\
	FP_a = \frac{x}{y}FP_b \\ 
	TP_a * FP_b = TP_b * FP_a
\end{cases} \Longleftrightarrow
\begin{cases}
	TP_a + FN_a = 1 - x \\
	TN_a + FP_a = x \\
	TP_b + FN_b = 1 - y \\
	TN_b + FP_b = y \\
	FP_a = \frac{x}{y}FP_b \\ 
	TP_a * FP_b = TP_b * \frac{x}{y} FP_b
\end{cases} \Longleftrightarrow
\begin{cases}
	FN_a = 1 - x - TP_a \\
	TN_a = x -  \frac{x}{y}FP_b \\
	FN_b = 1 - y - TP_b \\
	TN_b = y - FP_b \\
	FP_a = \frac{x}{y}FP_b \\ 
	TP_a = \frac{x}{y}TP_b \lor FP_b = 0
\end{cases}
$$

\subsubsection{Added constraints through inequalities}
\[
\frac{y}{x}(1-x) \ge TP_b
\]

\subsection{False positive parity and Predictive Equality}

$$ \begin{cases}
	TP_a + FN_a = 1 - x \\
	TN_a + FP_a = x \\
	TP_b + FN_b = 1 - y \\
	TN_b + FP_b = y \\
	FP_a = \frac{x}{y}FP_b \\ 
	FN_a*TN_b = FN_b*TN_a
\end{cases} \Longleftrightarrow
\begin{cases}
	TP_a + FN_a = 1 - x \\
	TN_a + FP_a = x \\
	TP_b + FN_b = 1 - y \\
	TN_b + FP_b = y \\
	x - TN_a = \frac{x}{y}(y-TN_b) \\ 
	FN_a * TN_b = FN_b * TN_a
\end{cases}
$$

$$ \Longleftrightarrow
\begin{cases}
	TP_a + FN_a = 1 - x \\
	TN_a + FP_a = x \\
	TP_b + FN_b = 1 - y \\
	TN_b + FP_b = y \\
	TN_a = \frac{x}{y} TN_b \\
	FN_a * TN_b = FN_b * \frac{x}{y} * TN_b
\end{cases} \Longleftrightarrow
\begin{cases}
	TP_a = 1 - x - FN_a \\
	FP_a = x - \frac{x}{y} TN_b  \\
	TP_b = 1 - y - FN_b \\
	FP_b = y - TN_b \\
	TN_a = \frac{x}{y} TN_b \\
	FN_a = \frac{x}{y} FN_b \lor TN_b=0
\end{cases}
$$

\subsubsection{Added constraints through inequalities}
\[
\frac{y}{x}(1-x)\ge FN_b
\]

\subsection{False positive parity and Overall accuracy Equality}

$$
\begin{cases}
	TP_a + FN_a = 1 - x \\
	TN_a + FP_a = x \\
	TP_b + FN_b = 1 - y \\
	TN_b + FP_b = y \\
	TN_a = \frac{x}{y}TN_b \\ 
	TP_a + TN_a = TP_b + TN_b
\end{cases} \Longleftrightarrow
\begin{cases}
	TP_a + FN_a = 1 - x \\
	TN_a + FP_a = x \\
	TP_b + FN_b = 1 - y \\
	TN_b + FP_b = y \\
	TN_a = \frac{x}{y}TN_b \\ 
	TP_a + \frac{x}{y}TN_b = TP_b + TN_b
\end{cases} \Longleftrightarrow
\begin{cases}
	TP_a + FN_a = 1 - x \\
	TN_a + FP_a = x \\
	TP_b + FN_b = 1 - y \\
	TN_b + FP_b = y \\
	TN_a = \frac{x}{y}TN_b \\ 
	\frac{x-y}{y}TN_b = TP_b - TP_a
\end{cases}
$$

$$ \Longleftrightarrow
\begin{cases}
	TP_a + FN_a = 1 - x \\
	TN_a + FP_a = x \\
	TP_b + FN_b = 1 - y \\
	TN_b + FP_b = y \\
	TN_b = \frac{y}{x-y} (TP_b - TPa) \\
	TN_a = \frac{x}{y} \frac{y}{x-y} (TP_b-TP_a)
\end{cases} \Longleftrightarrow
\begin{cases}
	FN_a = 1 - x - TP_a\\
	FP_a = x - \frac{x}{x-y} (TP_b-TP_a) \\
	FN_b = 1 - y - TP_b \\
	FP_b = y - \frac{y}{x-y} (TP_b - TPa) \\
	TN_b = \frac{y}{x-y} (TP_b - TPa) \\
	TN_a = \frac{x}{x-y} (TP_b-TP_a)
\end{cases}
$$

\subsubsection{Added constraints through inequalities}
\[
\begin{cases}
	x > y \\
	TP_a \le TP_b \\
	x - y \ge TP_b - TP_a \\
\end{cases} \lor
\begin{cases}
	x < y \\
	TP_a \ge TP_b \\
	x - y \le TP_b - TP_a \\
\end{cases}
\]

\subsection{False positive parity and Treatment Equality}
$$ 
\begin{cases}
	TP_a + FN_a = 1 - x \\
	TN_a + FP_a = x \\
	TP_b + FN_b = 1 - y \\
	TN_b + FP_b = y \\
	FP_a = \frac{x}{y}FP_b \\ 
	\frac{FN_a}{FP_a} = \frac{FN_b}{FP_b}
\end{cases} \Longleftrightarrow
\begin{cases}
	TP_a + FN_a = 1 - x \\
	TN_a + FP_a = x \\
	TP_b + FN_b = 1 - y \\
	TN_b + FP_b = y \\
	FP_a = \frac{x}{y}FP_b \\ 
	FN_a * FP_b = FN_b*\frac{x}{y}*FP_b
\end{cases} \Longleftrightarrow
\begin{cases}
	TP_a = 1 - x - FN_a  \\
	TN_a = x - \frac{x}{y}FP_b \\
	TP_b = 1 - y - FN_b \\
	TN_b = y - FP_b \\
	FP_a = \frac{x}{y}FP_b \\ 
	FN_a = \frac{x}{y}*FN_b \lor FP_b = 0
\end{cases}
$$

\subsubsection{Added constraints through inequalities}
\[
\frac{y}{x}(1-x) \ge FN_b
\]

\subsection{Predictive Parity and Predictive Equality}

$$ 
\begin{cases}
	TP_a + FN_a = 1 - x \\
	TN_a + FP_a = x \\
	TP_b + FN_b = 1 - y \\
	TN_b + FP_b = y \\
	TP_a = \frac{TP_b FP_a}{FP_b} \\
	FN_a = \frac{FN_b TN_a}{TN_b}
\end{cases}
$$

\subsection{Predictive Parity and Overall accuracy Equality}
$$ 
\begin{cases}
	TP_a + FN_a = 1 - x \\
	TN_a + FP_a = x \\
	TP_b + FN_b = 1 - y \\
	TN_b + FP_b = y \\
	TP_a * FP_b = TP_b * FP_a \\
	TP_a + TN_a = TP_b + TN_b
\end{cases} \Longleftrightarrow
\begin{cases}
	TP_a + FN_a = 1 - x \\
	TN_a + FP_a = x \\
	TP_b + FN_b = 1 - y \\
	TN_b + FP_b = y \\
	TP_a = \frac{TP_b * FP_a}{FP_b} \\
	\frac{TP_b * FP_a}{FP_b} + x - FP_a = TP_b + y - FP_b
\end{cases}
$$

$$ \Longleftrightarrow
\begin{cases}
	TP_a + FN_a = 1 - x \\
	TN_a + FP_a = x \\
	TP_b + FN_b = 1 - y \\
	TN_b + FP_b = y \\
	TP_a = \frac{TP_b * FP_a}{FP_b} \\
	\frac{FP_a - FP_b}{FP_b} TP_b = y - x + FP_a - FP_b
\end{cases} \Longleftrightarrow
\begin{cases}
	TP_a + FN_a = 1 - x \\
	TN_a + FP_a = x \\
	TP_b + FN_b = 1 - y \\
	TN_b + FP_b = y \\
	TP_a = \frac{TP_b * FP_a}{FP_b} \\
	TP_b = FP_b \frac{y-x}{FP_a - FP_b} + FP_b \lor (FP_a = FP_b \land x=y)
\end{cases}
$$
$$ \Longleftrightarrow
\begin{cases}
	FN_a = 1 - x - FP_a \frac{y-x}{FP_a - FP_b} - FP_a  \\
	TN_a = x - FP_a \\
	FN_b = 1 - y - FP_b \frac{y-x}{FP_a - FP_b} - FP_b \\
	TN_b = y - FP_b \\
	TP_b = FP_b \frac{y-x}{FP_a - FP_b} + FP_b \\
	TP_a = FP_a \frac{y-x}{FP_a - FP_b} + FP_a
\end{cases} \lor 
\begin{cases}
	FN_a = 1 - x - TP_b \\
	TN_a = x - FP_b \\
	FN_b = 1 - y - TP_b \\
	TN_b = y - FP_b \\
	TP_a = TP_b \\
	FP_a = FP_b \land x=y
\end{cases}
$$

\subsubsection{Added constraints through inequalities}
\[
\begin{cases}
	1-y \ge 1-x \\
	FP_a \ge FP_b \\
	\frac{1 - y - FP_a}{1 - x - FP_a} \ge \frac{FP_b}{FP_a} \\
	\frac{1 -y + FP_b}{1 - x + FP_b} \ge \frac{FP_b}{FP_a}
\end{cases} \lor
\begin{cases}
	1-y \le 1-x \\
	FP_a \le FP_b \\
	\frac{1 - y - FP_a}{1 - x - FP_a} \le \frac{FP_b}{FP_a} \\
	\frac{1 -y + FP_b}{1 - x + FP_b} \le \frac{FP_b}{FP_a}
\end{cases}
\]

\subsection{Predictive Parity and Treatment Equality}
$$
\begin{cases}
	TP_a + FN_a = 1 - x \\
	TN_a + FP_a = x \\
	TP_b + FN_b = 1 - y \\
	TN_b + FP_b = y \\
	TP_a * FP_b = TP_b * FP_a \\
	FN_a * FP_b = FN_b*FP_a
\end{cases} \Longleftrightarrow
\begin{cases}
	TP_a + FN_a = 1 - x \\
	TN_a + FP_a = x \\
	TP_b + FN_b = 1 - y \\
	TN_b + FP_b = y \\
	FN_a = \frac{FN_b*FP_a}{FP_b} \\
	(1 - x - \frac{FN_b*FP_a}{FP_b})*FP_b = (1-y-FN_b)*FP_a
\end{cases}
$$
$$ \Longleftrightarrow
\begin{cases}
	TP_a + FN_a = 1 - x \\
	TN_a + FP_a = x \\
	TP_b + FN_b = 1 - y \\
	TN_b + FP_b = y \\
	FN_a = \frac{FN_b*FP_a}{FP_b} \\
	FP_b -xFP_b - FN_b*FP_a = FP_a -yFP_a - FN_b*FP_a
\end{cases} \Longleftrightarrow
\begin{cases}
	TP_a + FN_a = 1 - x \\
	TN_a + FP_a = x \\
	TP_b + FN_b = 1 - y \\
	TN_b + FP_b = y \\
	FN_a = \frac{FN_b*FP_a}{FP_b} \\
	FP_b -xFP_b = FP_a -yFP_a
\end{cases} 
$$
$$ \Longleftrightarrow
\begin{cases}
	TP_a + FN_a = 1 - x \\
	TN_a + FP_a = x \\
	TP_b + FN_b = 1 - y \\
	TN_b + FP_b = y \\
	FN_a = \frac{FN_b*FP_a}{FP_b} \\
	FP_a = \frac{1-x}{1-y}FP_b
\end{cases} \Longleftrightarrow
\begin{cases}
	TP_a = 1 - x - \frac{1-x}{1-y}FN_b \\
	TN_a = x - \frac{1-x}{1-y}FP_b \\
	TP_b = 1 - y - FN_b \\
	TN_b = y - FP_b \\
	FN_a = \frac{1-x}{1-y}FN_b \\
	FP_a = \frac{1-x}{1-y}FP_b
\end{cases}
$$

\subsubsection{Added constraints through inequalities}
\[
\frac{1-y}{1-x}x \ge FP_b 
\]

\subsection{Predictive Equality and Overall accuracy Equality}
$$
\begin{cases}
	TP_a + FN_a = 1 - x \\
	TN_a + FP_a = x \\
	TP_b + FN_b = 1 - y \\
	TN_b + FP_b = y \\
	FN_a * TN_b = FN_b*TN_a \\
	TP_a + TN_a = TP_b + TN_b
\end{cases} \Longleftrightarrow
\begin{cases}
	TP_a + FN_a = 1 - x \\
	TN_a + FP_a = x \\
	TP_b + FN_b = 1 - y \\
	TN_b + FP_b = y \\
	FN_a = \frac{FN_b*TN_a}{TN_b} \\
	1 - x - FN_a + TN_a = 1 - y - FN_b + TN_b
\end{cases}
$$

$$ \Longleftrightarrow
\begin{cases}
	TP_a + FN_a = 1 - x \\
	TN_a + FP_a = x \\
	TP_b + FN_b = 1 - y \\
	TN_b + FP_b = y \\
	FN_a = \frac{FN_b*TN_a}{TN_b} \\
	\frac{TN_b - FN_b}{TN_b} TN_a = x - y +TN_b - FN_b
\end{cases} \Longleftrightarrow
\begin{cases}
	TP_a + FN_a = 1 - x \\
	TN_a + FP_a = x \\
	TP_b + FN_b = 1 - y \\
	TN_b + FP_b = y \\
	TN_a = TN_b \frac{x-y}{TN_b - FN_b} + TN_b \lor (TN_b = FN_b \land x=y) \\
	FN_a = \frac{FN_b*TN_a}{TN_b} 
\end{cases}
$$
$$
\Longleftrightarrow
\begin{cases}
	TP_a = 1 - x - FN_b\frac{x-y}{TN_b-FN_b} - FN_b\\
	FP_a = x - TN_b \frac{x-y}{TN_b - FN_b} - TN_b \\
	TP_b = 1 - y - FN_b\\
	FP_b = y - TN_b \\
	TN_a = TN_b \frac{x-y}{TN_b - FN_b} + TN_b \\
	FN_a = FN_b\frac{x-y}{TN_b-FN_b} + FN_b
\end{cases} \lor
\begin{cases}
	TP_a = 1 - x - TN_a \\
	FP_a = x - TN_a\\
	TP_b = 1 - y - FN_b \\
	FP_b = y - FN_b \\
	TN_b = FN_b \land x=y \\
	FN_a = TN_a
\end{cases}
$$

\subsubsection{Added constraints through inequalities}
\[
\begin{cases}
	TN_b \ge FN_b \\
	y \ge x \\
	\frac{x-TN_b}{y-TN_b} \le \frac{TN_b}{FN_b} \\
	\frac{1-x-FN_b}{1-y-FN_b} \ge \frac{FN_b}{TN_b}
\end{cases} \lor
\begin{cases}
	TN_b \le FN_b \\
	y \le x \\
	\frac{x-TN_b}{y-TN_b} \ge \frac{TN_b}{FN_b} \\
	\frac{1-x-FN_b}{1-y-FN_b} \le \frac{FN_b}{TN_b}
\end{cases}
\]

\subsection{Predictive Equality and Treatment Equality}
$$
\begin{cases}
	TP_a + FN_a = 1 - x \\
	TN_a + FP_a = x \\
	TP_b + FN_b = 1 - y \\
	TN_b + FP_b = y \\
	FN_a * TN_b = FN_b*TN_a \\
	\frac{FN_a}{FP_a} = \frac{FN_b}{FP_b}
\end{cases} \Longleftrightarrow
\begin{cases}
	TP_a + FN_a = 1 - x \\
	TN_a + FP_a = x \\
	TP_b + FN_b = 1 - y \\
	TN_b + FP_b = y \\
	FN_a = \frac{FN_b*TN_a}{TN_b} \\
	FN_a = \frac{FN_b*FP_a}{FP_b}
\end{cases} \Longleftrightarrow
\begin{cases}
	TP_a + FN_a = 1 - x \\
	TN_a + FP_a = x \\
	TP_b + FN_b = 1 - y \\
	TN_b + FP_b = y \\
	\frac{FN_b*FP_a}{FP_b} = \frac{FN_b*TN_a}{TN_b} \\
	FN_a = \frac{FN_b*TN_a}{TN_b}
\end{cases} 
$$

$$ \Longleftrightarrow
\begin{cases}
	TP_a + FN_a = 1 - x \\
	TN_a + FP_a = x \\
	TP_b + FN_b = 1 - y \\
	TN_b + FP_b = y \\
	TN_a*FP_b = TN_b*FP_a \lor FN_b=0\\
	FN_a = \frac{FN_b*TN_a}{TN_b}
\end{cases} \Longleftrightarrow
\begin{cases}
	TP_a + FN_a = 1 - x \\
	TN_a + FP_a = x \\
	TP_b + FN_b = 1 - y \\
	TN_b + FP_b = y \\
	TN_a(y - TN_b) = TN_b(x - TN_a) \lor FN_b=0 \\
	FN_a = \frac{FN_b*TN_a}{TN_b}
\end{cases}$$
$$ \Longleftrightarrow
\begin{cases}
	TP_a + FN_a = 1 - x \\
	TN_a + FP_a = x \\
	TP_b + FN_b = 1 - y \\
	TN_b + FP_b = y \\
	TN_a = \frac{x}{y} TN_b \lor FN_b=0 \\
	FN_a = \frac{FN_b*TN_a}{TN_b}
\end{cases} \Longleftrightarrow
\begin{cases}
	TP_a = 1 - x - \frac{x}{y} FN_b \\
	FP_a = x - \frac{x}{y} TN_b \\
	TP_b = 1 - y - FN_b \\
	FP_b = y - TN_b \\
	TN_a = \frac{x}{y} TN_b\\
	FN_a = \frac{x}{y} FN_b
\end{cases} \lor 
\begin{cases}
	TP_a = 1 - x \\
	TN_a = x - FP_a \\
	TP_b = 1 - y \\
	TN_b = y - FP_b \\
	FN_b=0 \\
	FN_a = 0
\end{cases}
$$

\subsubsection{Added constraints through inequalities}
\[ \frac{y}{x}(1-x) \ge TN_b
\]

\subsection{Overall accuracy Equality and Treatment Equality}
$$
\begin{cases}
	TP_a + FN_a = 1 - x \\
	TN_a + FP_a = x \\
	TP_b + FN_b = 1 - y \\
	TN_b + FP_b = y \\
	FP_a + FN_a = FP_b + FN_b \\
	FN_a = \frac{FN_b*FP_a}{FP_b}
\end{cases} \Longleftrightarrow
\begin{cases}
	TP_a + FN_a = 1 - x \\
	TN_a + FP_a = x \\
	TP_b + FN_b = 1 - y \\
	TN_b + FP_b = y \\
	FP_a + \frac{FN_b*FP_a}{FP_b} = FP_b + FN_b \\
	FN_a = \frac{FN_b*FP_a}{FP_b}
\end{cases} \Longleftrightarrow
\begin{cases}
	TP_a + FN_a = 1 - x \\
	TN_a + FP_a = x \\
	TP_b + FN_b = 1 - y \\
	TN_b + FP_b = y \\
	FP_a \frac{FP_b + FN_b}{FP_b} = FP_b + FN_b \\
	FN_a = \frac{FN_b*FP_a}{FP_b}
\end{cases} $$
$$ \Longleftrightarrow
\begin{cases}
	TP_a + FN_a = 1 - x \\
	TN_a + FP_a = x \\
	TP_b + FN_b = 1 - y \\
	TN_b + FP_b = y \\
	FP_a = FP_b \lor FP_b = FN_b = 0 \\
	FN_a = \frac{FN_b*FP_a}{FP_b}
\end{cases} \Longleftrightarrow
\begin{cases}
	TP_a = 1 - x - FN_b \\
	TN_a = x - FP_b \\
	TP_b = 1 - y - FN_b \\
	TN_b = y - FP_b \\
	FP_a = FP_b \\
	FN_a = FN_b
\end{cases} \lor
\begin{cases}
	TP_a = 1 - x \\
	TN_a = x -FP_a \\
	TP_b = 1 - y \\
	TN_b = y \\
	FP_b = FN_b = 0 \\
	FN_a = 0
\end{cases}
$$

\subsubsection{Added constraints through inequalities}
\[
\begin{cases}
	x \ge FP_b \\
	1 - x \ge FN_b
\end{cases}
\]

\section{Combining three definitions}
\label{sec:app_3_def}

\subsection{Demographic Parity and Equalised Odds}
\label{subsec:dem_par_eq_odds}
$$ \begin{cases}
	TP_a + FN_a = 1 - x \\
	TN_a + FP_a = x \\
	TP_b + FN_b = 1 - y \\
	TN_b + FP_b = y \\
	TP_a = \frac{1-x}{1-y}TP_b \\
	FP_a = \frac{x-y}{1-y}TP_b + FP_b \\
	TP_a = \frac{y-x}{y}FP_b + TP_b \\
	FP_a = \frac{x}{y} FP_b
\end{cases} \Longleftrightarrow
\begin{cases}
	TP_a + FN_a = 1 - x \\
	TN_a + FP_a = x \\
	TP_b + FN_b = 1 - y \\
	TN_b + FP_b = y \\
	TP_a = \frac{1-x}{1-y}TP_b \\
	FP_a = \frac{x}{y}FP_b \\
	\frac{x}{y}FP_b = \frac{x-y}{1-y}TP_b + FP_b \\
	\frac{1-x}{1-y}TP_b = \frac{y-x}{y}FP_b + TP_b
\end{cases} \Longleftrightarrow
\begin{cases}
	TP_a + FN_a = 1 - x \\
	TN_a + FP_a = x \\
	TP_b + FN_b = 1 - y \\
	TN_b + FP_b = y \\
	TP_a = \frac{1-x}{1-y}TP_b \\
	FP_a = \frac{x}{y}FP_b \\
	\frac{x-y}{y}FP_b = \frac{x-y}{1-y}TP_b    \\
	\frac{1-x-1+y}{1-y}TP_b = \frac{y-x}{y}FP_b
\end{cases}
$$

$$ \Longleftrightarrow 
\begin{cases}
	TP_a + FN_a = 1 - x \\
	TN_a + FP_a = x \\
	TP_b + FN_b = 1 - y \\
	TN_b + FP_b = y \\
	TP_a = \frac{1-x}{1-y}TP_b \\
	FP_a = \frac{x}{y}FP_b \\
	FP_b = \frac{y}{1-y}TP_b \lor x=y \\
	\frac{y-x}{1-y}TP_b = \frac{y-x}{y}FP_b
\end{cases} \Longleftrightarrow 
\begin{cases}
	TP_a + FN_a = 1 - x \\
	TN_a + FP_a = x \\
	TP_b + FN_b = 1 - y \\
	TN_b + FP_b = y \\
	TP_a = \frac{1-x}{1-y}TP_b \\
	FP_a = \frac{x}{y}FP_b \\
	FP_b = \frac{y}{1-y}TP_b \\
	\frac{y-x}{1-y}TP_b = \frac{y-x}{y}\frac{y}{1-y}TP_b
\end{cases} \lor
\begin{cases}
	TP_a + FN_a = 1 - x \\
	TN_a + FP_a = x \\
	TP_b + FN_b = 1 - y \\
	TN_b + FP_b = y \\
	TP_a = TP_b \\
	FP_a = FP_b \\
	x=y \\
	0 = 0
\end{cases}
$$
$$ \Longleftrightarrow
\begin{cases}
	FN_a = 1 - x - \frac{1-x}{1-y}TP_b  \\
	TN_a = x - \frac{x}{1-y}TP_b \\
	FN_b = 1 - y - TP_b \\
	TN_b = y - \frac{y}{1-y}TP_b \\
	TP_a = \frac{1-x}{1-y}TP_b \\
	FP_a = \frac{x}{1-y}TP_b \\
	FP_b = \frac{y}{1-y}TP_b \\
	TP_b = TP_b
\end{cases} \lor
\begin{cases}
	FN_a = 1 - x - TP_b \\
	TN_a = x - FP_b \\
	FN_b = 1 - y - TP_b \\
	TN_b = y - FP_b  \\
	TP_a = TP_b \\
	FP_a = FP_b \\
	x=y \\
	0 = 0
\end{cases} 
$$

\subsection{Demographic Parity, Equal Opportunity and Predictive Parity}
$$ \begin{cases}
	TP_a + FN_a = 1 - x \\
	TN_a + FP_a = x \\
	TP_b + FN_b = 1 - y \\
	TN_b + FP_b = y \\
	TP_a = \frac{1-x}{1-y}TP_b \\
	FP_a = \frac{x-y}{1-y}TP_b + FP_b \\
	TP_a = TP_b \\
	FP_a = FP_b 
\end{cases} \Longleftrightarrow
\begin{cases}
	TP_a + FN_a = 1 - x \\
	TN_a + FP_a = x \\
	TP_b + FN_b = 1 - y \\
	TN_b + FP_b = y \\
	TP_a = \frac{1-x}{1-y}TP_b \\
	FP_a = \frac{x-y}{1-y}TP_b + FP_a \\
	TP_a = TP_b \\
	FP_a = FP_b 
\end{cases} \Longleftrightarrow
\begin{cases}
	TP_a + FN_a = 1 - x \\
	TN_a + FP_a = x \\
	TP_b + FN_b = 1 - y \\
	TN_b + FP_b = y \\
	TP_a = \frac{1-x}{1-y}TP_b \\
	\frac{x-y}{1-y}TP_b = 0\\
	TP_a = TP_b \\
	FP_a = FP_b 
\end{cases}
$$

$$ \Longleftrightarrow
\begin{cases}
	TP_a + FN_a = 1 - x \\
	TN_a + FP_a = x \\
	TP_b + FN_b = 1 - y \\
	TN_b + FP_b = y \\
	TP_a = \frac{1-x}{1-y}TP_b \\
	x-y = 0 \lor TP_b = 0\\
	TP_a = TP_b \\
	FP_a = FP_b 
\end{cases} \Longleftrightarrow
\begin{cases}
	FN_a = 1 - x - TP_b \\
	TN_a = x - FP_b  \\
	FN_b = 1 - y - TP_b \\
	TN_b = y - FP_b \\
	x = y \\
	TP_a = TP_b \\
	TP_a = TP_b \\
	FP_a = FP_b 
\end{cases} \lor
\begin{cases}
	FN_a = 1 - x \\
	TN_a = x - FP_b  \\
	FN_b = 1 - y \\
	TN_b = y - FP_b \\
	TP_b = 0 \\
	TP_a = 0 \\
	TP_a = 0 \\
	FP_a = FP_b 
\end{cases}
$$

\subsection{Demographic Parity, Equal Opportunity and Predictive Equality}

$$
\begin{cases}
	TP_a + FN_a = 1 - x \\
	TN_a + FP_a = x \\
	TP_b + FN_b = 1 - y \\
	TN_b + FP_b = y \\
	TP_a = \frac{1-x}{1-y} TP_b \\
	FP_a = \frac{x-y}{1-y}TP_b + FP_b \\
	TN_a = TN_b \\
	FN_a = FN_b 
\end{cases} \Longleftrightarrow
\begin{cases}
	TP_a + FN_a = 1 - x \\
	TN_a + FP_a = x \\
	TP_b + FN_b = 1 - y \\
	TN_b + FP_b = y \\
	TP_a = \frac{1-x}{1-y} TP_b \\
	FP_a = \frac{x-y}{1-y}TP_b + FP_b \\
	-FP_a = - x + y - FP_b \\
	-TP_a = -1 + x + 1 - y - TP_b
\end{cases}
 \Longleftrightarrow
\begin{cases}
	TP_a + FN_a = 1 - x \\
	TN_a + FP_a = x \\
	TP_b + FN_b = 1 - y \\
	TN_b + FP_b = y \\
	FP_a = x - y + FP_b \\
	TP_a = - x + y + TP_b \\
	- x + y + TP_b = \frac{1-x}{1-y} TP_b \\
	x - y + FP_b = \frac{x-y}{1-y}TP_b + FP_b 
\end{cases} $$

$$ \Longleftrightarrow
\begin{cases}
	TP_a + FN_a = 1 - x \\
	TN_a + FP_a = x \\
	TP_b + FN_b = 1 - y \\
	TN_b + FP_b = y \\
	FP_a = x - y + FP_b \\
	TP_a = - x + y + TP_b \\
	- x + y = \frac{y-x}{1-y} TP_b \\
	x - y = \frac{x-y}{1-y}TP_b
\end{cases} \Longleftrightarrow
\begin{cases}
	TP_a + FN_a = 1 - x \\
	TN_a + FP_a = x \\
	TP_b + FN_b = 1 - y \\
	TN_b + FP_b = y \\
	FP_a = x - y + FP_b \\
	TP_a = - x + y + TP_b \\
	1 = \frac{1}{1-y} TP_b \lor x=y \\
	1 = \frac{1}{1-y} TP_b \lor x=y
\end{cases}
 \Longleftrightarrow
\begin{cases}
	FN_a = 0 \\
	TN_a = y - FP_b \\
	FN_b = 0 \\
	TN_b = y - FP_b \\
	FP_a = x - y + FP_b \\
	TP_a = 1 - x  \\
	TP_b = 1 - y
\end{cases} \lor
\begin{cases}
	FN_a = 1 - x - TP_b \\
	TN_a = x - FP_b \\
	FN_b = 1 - y - TP_b \\
	TN_b = y - FP_b \\
	FP_a = FP_b \\
	TP_a = TP_b \\
	x=y
\end{cases}
$$

\subsection{Demographic Parity, Equal Opportunity and Overall accuracy Equality}
$$
\begin{cases}
	TP_a + FN_a = 1 - x \\
	TN_a + FP_a = x \\
	TP_b + FN_b = 1 - y \\
	TN_b + FP_b = y \\
	TP_a = \frac{1-x}{1-y} TP_b \\
	FP_a = \frac{x-y}{1-y}TP_b + FP_b \\
	FP_a = FP_b + \frac{x-y}{2} \\
	TP_a = TP_b + \frac{y-x}{2}
\end{cases} \Longleftrightarrow
\begin{cases}
	TP_a + FN_a = 1 - x \\
	TN_a + FP_a = x \\
	TP_b + FN_b = 1 - y \\
	TN_b + FP_b = y \\
	TP_a = \frac{1-x}{1-y} TP_b \\
	FP_a = \frac{x-y}{1-y}TP_b + FP_b \\
	\frac{x-y}{1-y}TP_b + FP_b = FP_b + \frac{x-y}{2} \\
	\frac{1-x}{1-y} TP_b = TP_b + \frac{y-x}{2}
\end{cases} \Longleftrightarrow
\begin{cases}
	TP_a + FN_a = 1 - x \\
	TN_a + FP_a = x \\
	TP_b + FN_b = 1 - y \\
	TN_b + FP_b = y \\
	TP_a = \frac{1-x}{1-y} TP_b \\
	FP_a = \frac{x-y}{1-y}TP_b + FP_b \\
	\frac{x-y}{1-y} TP_b = \frac{x-y}{2} \\
	\frac{y-x}{1-y} TP_b = \frac{y-x}{2}
\end{cases}
$$
$$ \Longleftrightarrow
\begin{cases}
	TP_a + FN_a = 1 - x \\
	TN_a + FP_a = x \\
	TP_b + FN_b = 1 - y \\
	TN_b + FP_b = y \\
	TP_a = \frac{1-x}{1-y} TP_b \\
	FP_a = \frac{x-y}{1-y}TP_b + FP_b \\
	\frac{x-y}{1-y} TP_b = \frac{x-y}{2} 
\end{cases} \Longleftrightarrow
\begin{cases}
	TP_a + FN_a = 1 - x \\
	TN_a + FP_a = x \\
	TP_b + FN_b = 1 - y \\
	TN_b + FP_b = y \\
	TP_a = \frac{1-x}{1-y} TP_b \\
	FP_a = \frac{x-y}{1-y}TP_b + FP_b \\
	TP_b = \frac{1-y}{2} \lor x=y
\end{cases}
 \Longleftrightarrow
\begin{cases}
	FN_a = \frac{1-x}{2} \\
	TN_a = x - \frac{x-y}{2} - FP_b  \\
	FN_b = \frac{1-y}{2} \\
	TN_b = y - FP_b \\
	TP_a = \frac{1-x}{2} \\
	FP_a = \frac{x-y}{2} + FP_b \\
	TP_b = \frac{1-y}{2}
\end{cases} \lor
\begin{cases}
	FN_a = 1 - x - TP_b \\
	TN_a = x - FP_b \\
	FN_b = 1 - y - TP_b  \\
	TN_b = y - FP_b \\
	TP_a = TP_b \\
	FP_a = FP_b \\
	x=y
\end{cases}
$$

\subsection{Demographic Parity, Equal Opportunity and Treatment Equality}
$$
\begin{cases}
	TP_a + FN_a = 1 - x \\
	TN_a + FP_a = x \\
	TP_b + FN_b = 1 - y \\
	TN_b + FP_b = y \\
	TP_a = \frac{1-x}{1-y} TP_b \\
	FP_a = \frac{x-y}{1-y}TP_b + FP_b \\
	FN_a = FN_b (\frac{x-y}{FP_b-FN_b} + 1) \\
	FP_a = FP_b (\frac{x-y}{FP_b-FN_b} + 1)
\end{cases} \Longleftrightarrow
\begin{cases}
	TP_a + FN_a = 1 - x \\
	TN_a + FP_a = x \\
	TP_b + FN_b = 1 - y \\
	TN_b + FP_b = y \\
	FN_a = \frac{1-x}{1-y} FN_b \\
	FP_a = \frac{x-y}{1-y} (1 - y - FN_b) + FP_b \\
	FN_a = FN_b (\frac{x-y}{FP_b-FN_b} + 1) \\
	FP_a = FP_b (\frac{x-y}{FP_b-FN_b} + 1)
\end{cases}
$$
$$ \Longleftrightarrow
\begin{cases}
	TP_a + FN_a = 1 - x \\
	TN_a + FP_a = x \\
	TP_b + FN_b = 1 - y \\
	TN_b + FP_b = y \\
	FN_a = \frac{1-x}{1-y} FN_b \\
	FP_a = x - y - \frac{x-y}{1-y} FN_b + FP_b \\
	\frac{1-x}{1-y} FN_b = FN_b (\frac{x-y}{FP_b-FN_b} + 1) \\
	x - y - \frac{x-y}{1-y} FN_b + FP_b = FP_b (\frac{x-y}{FP_b-FN_b} + 1)
\end{cases} \Longleftrightarrow
\begin{cases}
	TP_a + FN_a = 1 - x \\
	TN_a + FP_a = x \\
	TP_b + FN_b = 1 - y \\
	TN_b + FP_b = y \\
	FN_a = \frac{1-x}{1-y} FN_b \\
	FP_a = x - y - \frac{x-y}{1-y} FN_b + FP_b \\
	\frac{1-x}{1-y}= (\frac{x-y}{FP_b-FN_b} + 1) \lor FN_b = 0 \\
	x - y - \frac{x-y}{1-y} FN_b= FP_b (\frac{x-y}{FP_b-FN_b})
\end{cases}
$$
$$\Longleftrightarrow
\begin{cases}
	TP_a + FN_a = 1 - x \\
	TN_a + FP_a = x \\
	TP_b + FN_b = 1 - y \\
	TN_b + FP_b = y \\
	FN_a = \frac{1-x}{1-y} FN_b \\
	FP_a = x - y - \frac{x-y}{1-y} FN_b + FP_b \\
	-\frac{1}{1-y}= \frac{1}{FP_b-FN_b} \lor FN_b = 0 \lor x=y \\
	(x-y)(1-\frac{1}{1-y}FN_b) = FP_b\frac{x-y}{FP_b-FN_b}	
\end{cases} \Longleftrightarrow
\begin{cases}
	TP_a + FN_a = 1 - x \\
	TN_a + FP_a = x \\
	TP_b + FN_b = 1 - y \\
	TN_b + FP_b = y \\
	FN_a = \frac{1-x}{1-y} FN_b \\
	FP_a = x - y - \frac{x-y}{1-y} FN_b + FP_b \\
	-1+y= FP_b-FN_b \lor FN_b = 0 \lor x=y \\
	1 - \frac{1}{1-y}FN_b = \frac{FP_b}{FP_b-FN_b} \lor x=y
\end{cases}
$$
$$ \Longleftrightarrow
\begin{cases}
	TP_a + FN_a = 1 - x \\
	TN_a + FP_a = x \\
	TP_b + FN_b = 1 - y \\
	TN_b + FP_b = y \\
	FN_a = \frac{1-x}{1-y} FN_b \\
	FP_a = x - y - \frac{x-y}{1-y} FN_b + FP_b \\
	-1+y= FP_b-FN_b \lor FN_b = 0 \lor x=y \\
	FP_b - FN_b - \frac{1}{1-y}(FP_b-FN_b)FN_b = FP_b \lor x=y
\end{cases} \Longleftrightarrow
\begin{cases}
	TP_a + FN_a = 1 - x \\
	TN_a + FP_a = x \\
	TP_b + FN_b = 1 - y \\
	TN_b + FP_b = y \\
	FN_a = \frac{1-x}{1-y} FN_b \\
	FP_a = x - y - \frac{x-y}{1-y} FN_b + FP_b \\
	-1+y= FP_b-FN_b \lor FN_b = 0 \lor x=y \\
	-1 - \frac{1}{1-y} (FP_b - FN_b) = 0\lor x=y \lor FN_b = 0
\end{cases}
$$

$$ \Longleftrightarrow
\begin{cases}
	TP_a + FN_a = 1 - x \\
	TN_a + FP_a = x \\
	TP_b + FN_b = 1 - y \\
	TN_b + FP_b = y \\
	FN_a = \frac{1-x}{1-y} FN_b \\
	FP_a = x - y - \frac{x-y}{1-y} FN_b + FP_b \\
	-1+y= FP_b-FN_b \lor FN_b = 0 \lor x=y \\
	FP_b - FN_b = -1+y\lor x=y \lor FN_b = 0
\end{cases} \Longleftrightarrow
\begin{cases}
	TP_a + FN_a = 1 - x \\
	TN_a + FP_a = x \\
	TP_b + FN_b = 1 - y \\
	TN_b + FP_b = y \\
	FN_a = \frac{1-x}{1-y} FN_b \\
	FP_a = x - y - \frac{x-y}{1-y} FN_b + FP_b \\
	FP_b = -1 + y + FN_b \lor FN_b = 0 \lor x=y 
\end{cases}
$$ 

$$ \Longleftrightarrow
\begin{cases}
	TP_a = 1 - x - \frac{1-x}{1-y} FN_b  \\
	TN_a = y + \frac{x-y}{1-y} FN_b \\
	FN_b = 1 - y \\
	TN_b = y \\
	FN_a = \frac{1-x}{1-y} FN_b \\
	FP_a = x - y - \frac{x-y}{1-y} FN_b \\
	FP_b = - TP_b = 0
\end{cases} \lor
\begin{cases}
	TP_a = 1 - x \\
	TN_a = y + FP_b \\
	TP_b = 1 - y \\
	TN_b = y - FP_b \\
	FN_a = 0 \\
	FP_a = x - y - FP_b \\
	FN_b = 0 
\end{cases} \lor
\begin{cases}
	TP_a = 1 - x - FN_b \\
	TN_a = x - FP_b \\
	TP_b = 1 - y - FN_b \\
	TN_b = y - FP_b \\
	FN_a = FN_b \\
	FP_a = FP_b \\
	x=y 
\end{cases} 
$$

\subsection{Demographic Parity, False positive parity and Predictive Parity}
$$
\begin{cases}
	TP_a + FN_a = 1 - x \\
	TN_a + FP_a = x \\
	TP_b + FN_b = 1 - y \\
	TN_b + FP_b = y \\
	TP_a = \frac{y-x}{y}FP_b + TP_b \\
	FP_a = \frac{x}{y} FP_b \\
	TP_a = TP_b \\
	FP_a = FP_b
\end{cases} \Longleftrightarrow
\begin{cases}
	TP_a + FN_a = 1 - x \\
	TN_a + FP_a = x \\
	TP_b + FN_b = 1 - y \\
	TN_b + FP_b = y \\
	TP_b = \frac{y-x}{y}FP_b + TP_b \\
	FP_b = \frac{x}{y} FP_b \\
	TP_a = TP_b \\
	FP_a = FP_b
\end{cases} \Longleftrightarrow
\begin{cases}
	TP_a + FN_a = 1 - x \\
	TN_a + FP_a = x \\
	TP_b + FN_b = 1 - y \\
	TN_b + FP_b = y \\
	0 = \frac{y-x}{y}FP_b \\
	1 = \frac{x}{y} \lor FP_b = 0 \\
	TP_a = TP_b \\
	FP_a = FP_b
\end{cases}
$$
$$ \Longleftrightarrow
\begin{cases}
	TP_a + FN_a = 1 - x \\
	TN_a + FP_a = x \\
	TP_b + FN_b = 1 - y \\
	TN_b + FP_b = y \\
	x = y \lor FP_b = 0 \\
	x = y \lor FP_b = 0 \\
	TP_a = TP_b \\
	FP_a = FP_b
\end{cases} \Longleftrightarrow
\begin{cases}
	FN_a = 1 - x - TP_b \\
	TN_a = x - FP_b \\
	FN_b = 1 - y - TP_b \\
	TN_b = y - FP_b \\
	x = y \\
	TP_a = TP_b \\
	FP_a = FP_b
\end{cases} \lor
\begin{cases}
	FN_a = 1 - x - TP_b \\
	TN_a = x \\
	FN_b = 1 - y - TP_b \\
	TN_b = y \\
	FP_b = 0 \\
	TP_a = TP_b \\
	FP_a = 0
\end{cases}
$$

\subsection{Demographic Parity, False positive parity and Predictive Equality}
$$
\begin{cases}
	TP_a + FN_a = 1 - x \\
	TN_a + FP_a = x \\
	TP_b + FN_b = 1 - y \\
	TN_b + FP_b = y \\
	TP_a = \frac{y-x}{y}FP_b + TP_b \\
	TN_a = \frac{x}{y} TN_b \\
	TN_a = TN_b \\ 
	FN_a = FN_b
\end{cases} \Longleftrightarrow
\begin{cases}
	TP_a + FN_a = 1 - x \\
	TN_a + FP_a = x \\
	TP_b + FN_b = 1 - y \\
	TN_b + FP_b = y \\
	TP_a = \frac{y-x}{y}FP_b + TP_b \\
	TN_b = \frac{x}{y} TN_b \\
	TN_a = TN_b \\ 
	-TP_a = -1 + x + 1 - y - TP_b
\end{cases}
 \Longleftrightarrow
\begin{cases}
	TP_a + FN_a = 1 - x \\
	TN_a + FP_a = x \\
	TP_b + FN_b = 1 - y \\
	TN_b + FP_b = y \\
	TP_a = \frac{y-x}{y}FP_b + TP_b \\
	1 = \frac{x}{y} \lor TN_b = 0 \\
	TN_a = TN_b \\ 
	TP_a = - x + y + TP_b
\end{cases}
$$
$$ \Longleftrightarrow
\begin{cases}
	TP_a + FN_a = 1 - x \\
	TN_a + FP_a = x \\
	TP_b + FN_b = 1 - y \\
	TN_b + FP_b = y \\
	- x + y + TP_b = \frac{y-x}{y}FP_b + TP_b \\
	x = y \lor TN_b = 0 \\
	TN_a = TN_b \\ 
	TP_a = - x + y + TP_b
\end{cases}
\Longleftrightarrow
\begin{cases}
	TP_a + FN_a = 1 - x \\
	TN_a + FP_a = x \\
	TP_b + FN_b = 1 - y \\
	TN_b + FP_b = y \\
	- x + y  = \frac{y-x}{y}FP_b\\
	x = y \lor TN_b = 0 \\
	TN_a = TN_b \\ 
	TP_a = - x + y + TP_b
\end{cases} \Longleftrightarrow
\begin{cases}
	TP_a + FN_a = 1 - x \\
	TN_a + FP_a = x \\
	TP_b + FN_b = 1 - y \\
	TN_b + FP_b = y \\
	1 = \frac{1}{y}FP_b \lor x=y\\
	x = y \lor TN_b = 0 \\
	TN_a = TN_b \\ 
	TP_a = - x + y + TP_b
\end{cases}
$$
$$  \Longleftrightarrow
\begin{cases}
	TP_a + FN_a = 1 - x \\
	TN_a + FP_a = x \\
	TP_b + FN_b = 1 - y \\
	TN_b + FP_b = y \\
	y = FP_b \lor x=y\\
	x = y \lor TN_b = 0 \\
	TN_a = TN_b \\ 
	TP_a = - x + y + TP_b
\end{cases}
\Longleftrightarrow
\begin{cases}
	FN_a = 1 - x - TP_b \\
	FP_a = x - TN_b \\
	FN_b = 1 - y - TP_b \\
	FP_b = y - TN_b \\
	x = y \\
	x = y\\
	TN_a = TN_b \\ 
	TP_a = TP_b
\end{cases} \lor
\begin{cases}
	FN_a = 1 - y - TP_b \\
	FP_a = x \\
	FN_b = 1 - y - TP_b \\
	0 = 0 \\
	FP_b = y  \\
	TN_b = 0 \\
	TN_a = 0 \\ 
	TP_a = - x + y + TP_b
\end{cases}
$$

\subsection{Demographic Parity, False positive parity and Overall accuracy Equality}
$$
\begin{cases}
	TP_a + FN_a = 1 - x \\
	TN_a + FP_a = x \\
	TP_b + FN_b = 1 - y \\
	TN_b + FP_b = y \\
	TP_a = \frac{y-x}{y}FP_b + TP_b \\
	FP_a = \frac{x}{y} FP_b \\
	TP_a = TP_b + \frac{y-x}{2} \\
	FP_a = FP_b + \frac{x-y}{2}
\end{cases} \Longleftrightarrow
\begin{cases}
	TP_a + FN_a = 1 - x \\
	TN_a + FP_a = x \\
	TP_b + FN_b = 1 - y \\
	TN_b + FP_b = y \\
	TP_a = \frac{y-x}{y}FP_b + TP_b \\
	FP_a = \frac{x}{y} FP_b \\
	\frac{y-x}{y}FP_b + TP_b = TP_b + \frac{y-x}{2} \\
	\frac{x}{y} FP_b = FP_b + \frac{x-y}{2}
\end{cases} \Longleftrightarrow
\begin{cases}
	TP_a + FN_a = 1 - x \\
	TN_a + FP_a = x \\
	TP_b + FN_b = 1 - y \\
	TN_b + FP_b = y \\
	TP_a = \frac{y-x}{y}FP_b + TP_b \\
	FP_a = \frac{x}{y} FP_b \\
	\frac{y-x}{y}FP_b = \frac{y-x}{2} \\
	\frac{x-y}{y} FP_b = \frac{x-y}{2}
\end{cases}
$$

$$ \Longleftrightarrow
\begin{cases}
	TP_a + FN_a = 1 - x \\
	TN_a + FP_a = x \\
	TP_b + FN_b = 1 - y \\
	TN_b + FP_b = y \\
	TP_a = \frac{y-x}{y}FP_b + TP_b \\
	FP_a = \frac{x}{y} FP_b \\
	\frac{1}{y}FP_b = \frac{1}{2} \lor x=y \\
\end{cases} \Longleftrightarrow
\begin{cases}
	TP_a + FN_a = 1 - x \\
	TN_a + FP_a = x \\
	TP_b + FN_b = 1 - y \\
	TN_b + FP_b = y \\
	TP_a = \frac{y-x}{y}FP_b + TP_b \\
	FP_a = \frac{x}{y} FP_b \\
	FP_b = \frac{y}{2} \lor x=y 
\end{cases}
$$
$$ \Longleftrightarrow
\begin{cases}
	FN_a = 1 - x - \frac{y-x}{2} - TP_b \\
	TN_a = \frac{x}{2} \\
	FN_b = 1 - y - TP_b \\
	TN_b = \frac{y}{2}  \\
	TP_a = \frac{y-x}{2}+ TP_b \\
	FP_a = \frac{x}{2}\\
	FP_b = \frac{y}{2} 
\end{cases} \lor 
\begin{cases}
	FN_a = 1 - x - TP_b \\
	TN_a = x - FP_b \\
	FN_b = 1 - y - TP_b \\
	TN_b = y - FP_b \\
	TP_a = TP_b \\
	FP_a = FP_b \\
	x=y
\end{cases}
$$

\subsection{Demographic Parity, False positive parity and Treatment Equality}

$$ \begin{cases}
	TP_a + FN_a = 1 - x \\
	TN_a + FP_a = x \\
	TP_b + FN_b = 1 - y \\
	TN_b + FP_b = y \\
	TP_a = \frac{y-x}{y}FP_b + TP_b \\
	FP_a = \frac{x}{y}FP_b \\
	FN_a = FN_b*(\frac{x-y}{FP_b-FN_b} + 1) \\
	FP_a = FP_b*(\frac{x-y}{FP_b-FN_b} + 1)
\end{cases} \Longleftrightarrow
\begin{cases}
	TP_a + FN_a = 1 - x \\
	TN_a + FP_a = x \\
	TP_b + FN_b = 1 - y \\
	TN_b + FP_b = y \\
	-FN_a = -1 +x + 1 - y - FN_b + \frac{y-x}{y}FP_b \\
	FP_a = \frac{x}{y}FP_b \\
	\frac{x}{y}FP_b = FP_b*(\frac{x-y}{FP_b-FN_b} + 1) \\
	FN_a = FN_b*(\frac{x-y}{FP_b-FN_b} + 1) 
\end{cases}
$$
$$\Longleftrightarrow
\begin{cases}
	TP_a + FN_a = 1 - x \\
	TN_a + FP_a = x \\
	TP_b + FN_b = 1 - y \\
	TN_b + FP_b = y \\
	FN_a = y - x + FN_b - \frac{y-x}{y}FP_b \\
	FP_a = \frac{x}{y}FP_b \\
	\frac{x}{y} = (\frac{x-y}{FP_b-FN_b} + 1) \lor FP_b = 0 \\
	y - x + FN_b - \frac{y-x}{y}FP_b = FN_b*(\frac{x-y}{FP_b-FN_b} + 1) 
\end{cases} \Longleftrightarrow
\begin{cases}
	TP_a + FN_a = 1 - x \\
	TN_a + FP_a = x \\
	TP_b + FN_b = 1 - y \\
	TN_b + FP_b = y \\
	FN_a = y - x + FN_b - \frac{y-x}{y}FP_b \\
	FP_a = \frac{x}{y}FP_b \\
	\frac{x-y}{y} = \frac{x-y}{FP_b-FN_b} \lor FP_b = 0 \\
	y - x - \frac{y-x}{y}FP_b = FN_b*(\frac{x-y}{FP_b-FN_b}) 
\end{cases}
$$
$$ \Longleftrightarrow
\begin{cases}
	TP_a + FN_a = 1 - x \\
	TN_a + FP_a = x \\
	TP_b + FN_b = 1 - y \\
	TN_b + FP_b = y \\
	FN_a = y - x + FN_b - \frac{y-x}{y}FP_b \\
	FP_a = \frac{x}{y}FP_b \\
	FP_b - FN_b = y \lor FP_b = 0 \lor x-y = 0\\
	1 - \frac{1}{y}FP_b = - FN_b * \frac{1}{FP_b - FN_b} \lor x=y 
\end{cases} \Longleftrightarrow
\begin{cases}
	TP_a + FN_a = 1 - x \\
	TN_a + FP_a = x \\
	TP_b + FN_b = 1 - y \\
	TN_b + FP_b = y \\
	FN_a = y - x + FN_b - \frac{y-x}{y}FP_b \\
	FP_a = \frac{x}{y}FP_b \\
	FP_b - FN_b = y \lor FP_b = 0 \lor x-y = 0\\
	y - FP_b = - yFN_b * \frac{1}{FP_b - FN_b} \lor x=y 
\end{cases}
$$
$$\Longleftrightarrow
\begin{cases}
	TP_a + FN_a = 1 - x \\
	TN_a + FP_a = x \\
	TP_b + FN_b = 1 - y \\
	TN_b + FP_b = y \\
	FN_a = y - x + FN_b - \frac{y-x}{y}FP_b \\
	FP_a = \frac{x}{y}FP_b \\
	FP_b - FN_b = y \lor FP_b = 0 \lor x-y = 0\\
	yFP_b - yFN_b - FP_b^2 + FN_bFP_b = - yFN_b \lor x=y 
\end{cases} \Longleftrightarrow
\begin{cases}
	TP_a + FN_a = 1 - x \\
	TN_a + FP_a = x \\
	TP_b + FN_b = 1 - y \\
	TN_b + FP_b = y \\
	FN_a = y - x + FN_b - \frac{y-x}{y}FP_b \\
	FP_a = \frac{x}{y}FP_b \\
	FP_b - FN_b = y \lor FP_b = 0 \lor x-y = 0\\
	yFP_b - FP_b^2 + FN_bFP_b = 0 \lor x=y 
\end{cases}$$
$$ \Longleftrightarrow
\begin{cases}
	TP_a + FN_a = 1 - x \\
	TN_a + FP_a = x \\
	TP_b + FN_b = 1 - y \\
	TN_b + FP_b = y \\
	FN_a = y - x + FN_b - \frac{y-x}{y}FP_b \\
	FP_a = \frac{x}{y}FP_b \\
	FP_b - FN_b = y \lor FP_b = 0 \lor x = y\\
	y- FP_b + FN_b= 0 \lor x=y \lor FP_b = 0
\end{cases} \Longleftrightarrow
\begin{cases}
	TP_a + FN_a = 1 - x \\
	TN_a + FP_a = x \\
	TP_b + FN_b = 1 - y \\
	TN_b + FP_b = y \\
	FN_a = y - x + FN_b - \frac{y-x}{y}FP_b \\
	FP_a = \frac{x}{y}FP_b \\
	FP_b - FN_b = y \lor FP_b = 0 \lor x = y\\
	FP_b - FN_b = y \lor x=y \lor FP_b = 0
\end{cases}$$
$$ \Longleftrightarrow
\begin{cases}
	TP_a + FN_a = 1 - x \\
	TN_a + FP_a = x \\
	TP_b + FN_b = 1 - y \\
	TN_b + FP_b = y \\
	FN_a = y - x + FN_b - \frac{y-x}{y}FP_b \\
	FP_a = \frac{x}{y}FP_b \\
	FN_b = FP_b - y \lor FP_b = 0 \lor x = y
\end{cases} \Longleftrightarrow
\begin{cases}
	TP_a = 1 - x - FN_a \\
	TN_a = x - FP_a \\
	TP_b = 1 - y - FN_b \\
	TN_b = y - FP_b \\
	FN_b = FP_b - y \\ 
	FN_a = - x (1 - \frac{1}{y}FP_b) \\
	FP_a = \frac{x}{y}FP_b \\
\end{cases} \lor 
\begin{cases}
	TP_a + FN_a = 1 - x \\
	TN_a + FP_a = x \\
	TP_b + FN_b = 1 - y \\
	TN_b + FP_b = y \\
	FN_a = y - x + FN_b - \frac{y-x}{y}FP_b \\
	FP_a = \frac{x}{y}FP_b \\
	FP_a = 0 \lor FP_a = FP_b\\
\end{cases}
$$
$$\Longleftrightarrow
\begin{cases}
	TP_a = 1 - x \\
	TN_a = 0 \\
	TP_b = 1 - y \\
	TN_b = 0\\
	FN_b = 0 \\ 
	FN_a = 0 , FP_b = y \\
	FP_a = x \\
\end{cases} \lor 
\begin{cases}
	TP_a = 1 - y - FN_b \\
	TN_a = x \\
	TP_b = 1 - y - FN_b \\
	TN_b = y \\
	FP_b = 0 \\ 
	FN_a = y - x + FN_b \\
	FP_a = 0 \\
\end{cases} \lor
\begin{cases}
	TP_a = 1 - x - FN_b \\
	TN_a = x - FP_b \\
	TP_b = 1 - y - FN_b \\
	TN_b = y - FP_b \\
	x = y \\
	FN_a = FN_b\\
	FP_a = FP_b \\
\end{cases}
$$

\subsection{Demographic Parity, Predictive Parity and Predictive Equality}

$$
\begin{cases}
	TP_a + FN_a = 1 - x \\
	TN_a + FP_a = x \\
	TP_b + FN_b = 1 - y \\
	TN_b + FP_b = y \\
	TP_a = TP_b \\
	FP_a = FP_b \\
	TN_a = TN_b \\
	FN_a = FN_b
\end{cases} \Longleftrightarrow
\begin{cases}
	TP_a + FN_a = 1 - x \\
	TN_a + FP_a = x \\
	TP_b + FN_b = 1 - y \\
	TN_b + FP_b = y \\
	TP_a = TP_b \\
	FP_a = FP_b \\
	x - FP_a = y - FP_b \\
	1 - x - TP_a = 1 - y - TP_b
\end{cases} \Longleftrightarrow
\begin{cases}
	TP_a + FN_a = 1 - x \\
	TN_a + FP_a = x \\
	TP_b + FN_b = 1 - y \\
	TN_b + FP_b = y \\
	TP_a = TP_b \\
	FP_a = FP_b \\
	x - FP_a = y - FP_a \\
	- x - TP_a = - y - TP_a
\end{cases}
\Longleftrightarrow
\begin{cases}
	FN_a = 1 - x - TP_b \\
	TN_a = x - FP_b \\
	FN_b = 1 - y - TP_b \\
	TN_b = y - FP_b \\
	TP_a = TP_b \\
	FP_a = FP_b \\
	x = y \\
	x = y
\end{cases}
$$

\subsection{Demographic Parity, Predictive Parity and Overall Accuracy equality}
$$
\begin{cases}
	TP_a + FN_a = 1 - x \\
	TN_a + FP_a = x \\
	TP_b + FN_b = 1 - y \\
	TN_b + FP_b = y \\
	TP_a = TP_b \\
	FP_a = FP_b \\
	FP_a = FP_b + \frac{x-y}{2} \\
	TP_a = TP_b + \frac{y-x}{2}
\end{cases} \Longleftrightarrow
\begin{cases}
	TP_a + FN_a = 1 - x \\
	TN_a + FP_a = x \\
	TP_b + FN_b = 1 - y \\
	TN_b + FP_b = y \\
	TP_a = TP_b \\
	FP_a = FP_b \\
	FP_a = FP_a + \frac{x-y}{2} \\
	TP_a = TP_a + \frac{y-x}{2}
\end{cases} \Longleftrightarrow
\begin{cases}
	TP_a + FN_a = 1 - x \\
	TN_a + FP_a = x \\
	TP_b + FN_b = 1 - y \\
	TN_b + FP_b = y \\
	TP_a = TP_b \\
	FP_a = FP_b \\
	0 = \frac{x-y}{2} \\
	0 = \frac{y-x}{2}
\end{cases}
\Longleftrightarrow
\begin{cases}
	FN_a = 1 - x - TP_b \\
	TN_a = x - FP_b \\
	FN_b = 1 - y - TP_b \\
	TN_b = y - FP_b \\
	TP_a = TP_b \\
	FP_a = FP_b \\
	x = y \\
	x = y
\end{cases}
$$

\subsection{Demographic Parity, Predictive Parity and Treatment Equality}
$$
\begin{cases}
	TP_a + FN_a = 1 - x \\
	TN_a + FP_a = x \\
	TP_b + FN_b = 1 - y \\
	TN_b + FP_b = y \\
	TP_a = TP_b \\
	FP_a = FP_b \\
	FN_a = FN_b*(\frac{x-y}{FP_b-FN_b} + 1) \\
	FP_a = FP_b*(\frac{x-y}{FP_b-FN_b} + 1)
\end{cases} \Longleftrightarrow
\begin{cases}
	TP_a + FN_a = 1 - x \\
	TN_a + FP_a = x \\
	TP_b + FN_b = 1 - y \\
	TN_b + FP_b = y \\
	1-x-FN_a = 1-y-FN_b \\
	FP_a = FP_b \\
	FN_a = FN_b*(\frac{x-y}{FP_b-FN_b} + 1) \\
	FP_b = FP_b*(\frac{x-y}{FP_b-FN_b} + 1)
\end{cases}
\Longleftrightarrow
\begin{cases}
	TP_a + FN_a = 1 - x \\
	TN_a + FP_a = x \\
	TP_b + FN_b = 1 - y \\
	TN_b + FP_b = y \\
	FN_a = y-x+FN_b \\
	FP_a = FP_b \\
	y-x+FN_b = FN_b*(\frac{x-y}{FP_b-FN_b} + 1) \\
	1 = \frac{x-y}{FP_b-FN_b}+1 \lor FP_b = 0
\end{cases} 
$$
$$\Longleftrightarrow
\begin{cases}
	TP_a + FN_a = 1 - x \\
	TN_a + FP_a = x \\
	TP_b + FN_b = 1 - y \\
	TN_b + FP_b = y \\
	FN_a = y-x+FN_b \\
	FP_a = FP_b \\
	y-x = FN_b*(\frac{x-y}{FP_b-FN_b}) \\
	x=y \lor FP_b = 0
\end{cases}
\Longleftrightarrow
\begin{cases}
	TP_a = 1 - x - FN_b \\
	TN_a = x - FP_b\\
	TP_b = 1 - y - FN_b \\
	TN_b = y - FP_b \\
	FN_a = FN_b \\
	FP_a = FP_b \\
	0 = 0 \\
	x=y
\end{cases} \lor
\begin{cases}
	TP_a = 1 - y - FN_b \\
	TN_a = x\\
	TP_b = 1 - y - FN_b \\
	TN_b = y \\
	FN_a = y-x+FN_b \\
	FP_a = FP_b = 0 \\
	y-x = -(x-y) \\
	FP_b = 0
\end{cases}
$$

\subsection{Demographic Parity, Predictive Equality and Overall accuracy Equality}
$$
\begin{cases}
	TP_a + FN_a = 1 - x \\
	TN_a + FP_a = x \\
	TP_b + FN_b = 1 - y \\
	TN_b + FP_b = y \\
	TN_a = TN_b \\
	FN_a = FN_b \\
	FP_a = FP_b + \frac{x-y}{2} \\
	TP_a = TP_b + \frac{y-x}{2}
\end{cases} \Longleftrightarrow 
\begin{cases}
	TP_a + FN_a = 1 - x \\
	TN_a + FP_a = x \\
	TP_b + FN_b = 1 - y \\
	TN_b + FP_b = y \\
	x - FP_a = y - FP_b \\
	1 - x - TP_a = 1 - y - TP_b \\
	FP_a = FP_b + \frac{x-y}{2} \\
	TP_a = TP_b + \frac{y-x}{2}
\end{cases} \Longleftrightarrow 
\begin{cases}
	TP_a + FN_a = 1 - x \\
	TN_a + FP_a = x \\
	TP_b + FN_b = 1 - y \\
	TN_b + FP_b = y \\
	FP_a = x - y + FP_b \\
	TP_a = y -x + TP_b \\
	x - y + FP_b = FP_b + \frac{x-y}{2} \\
	y -x + TP_b = TP_b + \frac{y-x}{2}
\end{cases}
$$
$$ \Longleftrightarrow 
\begin{cases}
	TP_a + FN_a = 1 - x \\
	TN_a + FP_a = x \\
	TP_b + FN_b = 1 - y \\
	TN_b + FP_b = y \\
	FP_a = x - y + FP_b \\
	TP_a = y -x + TP_b \\
	x - y = \frac{x-y}{2} \\
	y -x = \frac{y-x}{2}
\end{cases} \Longleftrightarrow 
\begin{cases}
	TP_a + FN_a = 1 - x \\
	TN_a + FP_a = x \\
	TP_b + FN_b = 1 - y \\
	TN_b + FP_b = y \\
	FP_a = x - y + FP_b \\
	TP_a = y -x + TP_b \\
	x = y \\
	x = y
\end{cases} \Longleftrightarrow 
\begin{cases}
	FN_a = 1 - x - TP_b \\
	TN_a = x - FP_b \\
	FN_b = 1 - y - TP_b \\
	TN_b = y - FP_b \\
	FP_a = FP_b \\
	TP_a = TP_b \\
	x = y 
\end{cases}
$$

\subsection{Demographic Parity, Predictive Equality and Treatment Equality}
$$
\begin{cases}
	TP_a + FN_a = 1 - x \\
	TN_a + FP_a = x \\
	TP_b + FN_b = 1 - y \\
	TN_b + FP_b = y \\
	FP_a = x - y + FP_b \\
	FN_a = FN_b \\
	FN_a = FN_b (\frac{x-y}{FP_b - FN_b} + 1) \\
	FP_a = FP_b (\frac{x-y}{FP_b - FN_b} + 1)
\end{cases} \Longleftrightarrow
\begin{cases}
	TP_a + FN_a = 1 - x \\
	TN_a + FP_a = x \\
	TP_b + FN_b = 1 - y \\
	TN_b + FP_b = y \\
	FP_a = x - y + FP_b \\
	FN_a = FN_b \\
	FN_b = FN_b (\frac{x-y}{FP_b - FN_b} + 1) \\
	x - y + FP_b = FP_b (\frac{x-y}{FP_b - FN_b} + 1)
\end{cases}
\Longleftrightarrow
\begin{cases}
	TP_a + FN_a = 1 - x \\
	TN_a + FP_a = x \\
	TP_b + FN_b = 1 - y \\
	TN_b + FP_b = y \\
	FP_a = x - y + FP_b \\
	FN_a = FN_b \\
	1 = \frac{x-y}{FP_b - FN_b} + 1 \lor FN_b = 0 \\
	x - y = FP_b \frac{x-y}{FP_b - FN_b}
\end{cases} 
$$
$$ \Longleftrightarrow
\begin{cases}
	TP_a + FN_a = 1 - x \\
	TN_a + FP_a = x \\
	TP_b + FN_b = 1 - y \\
	TN_b + FP_b = y \\
	FP_a = x - y + FP_b \\
	FN_a = FN_b \\
	0 = \frac{x-y}{FP_b - FN_b} \lor FN_b = 0 \\
	1 = \frac{FP_b}{FP_b - FN_b} \lor x = y
\end{cases} \Longleftrightarrow
\begin{cases}
	TP_a + FN_a = 1 - x \\
	TN_a + FP_a = x \\
	TP_b + FN_b = 1 - y \\
	TN_b + FP_b = y \\
	FP_a = x - y + FP_b \\
	FN_a = FN_b \\
	x = y \lor FN_b = 0 \\
	FP_b - FN_b = FP_b \lor x = y
\end{cases} \Longleftrightarrow
\begin{cases}
	TP_a + FN_a = 1 - x \\
	TN_a + FP_a = x \\
	TP_b + FN_b = 1 - y \\
	TN_b + FP_b = y \\
	FP_a = x - y + FP_b \\
	FN_a = FN_b \\
	x = y \lor FN_b = 0 \\
	FN_b = 0 \lor x = y
\end{cases}
$$
$$\Longleftrightarrow
\begin{cases}
	TP_a = 1 - x - FN_b \\
	TN_a = x - FP_b \\
	TP_b = 1 - y - FN_b \\
	TN_b = y - FP_b \\
	FP_a = FP_b \\
	FN_a = FN_b \\
	x = y 
\end{cases} \lor
\begin{cases}
	TP_a = 1 - x \\
	TN_a = y - FP_b \\
	TP_b = 1 - y \\
	TN_b = y - FP_b \\
	FP_a = x - y + FP_b \\
	FN_a = 0 \\
	FN_b = 0 
\end{cases} 
$$

\subsection{Demographic Parity, Overall accuracy Equality and Treatment Equality}
$$
\begin{cases}
	TP_a + FN_a = 1 - x \\
	TN_a + FP_a = x \\
	TP_b + FN_b = 1 - y \\
	TN_b + FP_b = y \\
	FP_a = FP_b + \frac{x-y}{2} \\	
	TP_a = TP_b + \frac{y-x}{2} \\
	FN_a = FN_b (\frac{x-y}{FP_b - FN_b} + 1) \\
	FP_a = FP_b (\frac{x-y}{FP_b - FN_b} + 1)
\end{cases} \Longleftrightarrow
\begin{cases}
	TP_a + FN_a = 1 - x \\
	TN_a + FP_a = x \\
	TP_b + FN_b = 1 - y \\
	TN_b + FP_b = y \\
	FP_a = FP_b + \frac{x-y}{2} \\	
	FN_a = y - x + FN_b - \frac{y-x}{2} \\
	y - x + FN_b - \frac{y-x}{2} = FN_b (\frac{x-y}{FP_b - FN_b} + 1) \\
	FP_b + \frac{x-y}{2} = FP_b (\frac{x-y}{FP_b - FN_b} + 1)
\end{cases}
$$
$$\Longleftrightarrow
\begin{cases}
	TP_a + FN_a = 1 - x \\
	TN_a + FP_a = x \\
	TP_b + FN_b = 1 - y \\
	TN_b + FP_b = y \\
	FP_a = FP_b + \frac{x-y}{2} \\	
	FN_a = y - x + FN_b - \frac{y-x}{2} \\
	\frac{y-x}{2} = FN_b \frac{x-y}{FP_b - FN_b} \\
	\frac{x-y}{2} = FP_b \frac{x-y}{FP_b - FN_b}
\end{cases} \Longleftrightarrow
\begin{cases}
	TP_a + FN_a = 1 - x \\
	TN_a + FP_a = x \\
	TP_b + FN_b = 1 - y \\
	TN_b + FP_b = y \\
	FP_a = FP_b + \frac{x-y}{2} \\	
	FN_a = y - x + FN_b - \frac{y-x}{2} \\
	- \frac{1}{2} = FN_b \frac{1}{FP_b - FN_b} \lor x=y \\
	\frac{1}{2} = FP_b \frac{1}{FP_b - FN_b} \lor x=y
\end{cases}
$$
$$ \Longleftrightarrow
\begin{cases}
	TP_a + FN_a = 1 - x \\
	TN_a + FP_a = x \\
	TP_b + FN_b = 1 - y \\
	TN_b + FP_b = y \\
	FP_a = FP_b + \frac{x-y}{2} \\	
	FN_a = y - x + FN_b - \frac{y-x}{2} \\
	\frac{-FP_b + FN_b}{2} = FN_b \lor x=y \\
	\frac{FP_b - FN_b}{2} = FP_b \lor x=y
\end{cases}
\Longleftrightarrow
\begin{cases}
	TP_a + FN_a = 1 - x \\
	TN_a + FP_a = x \\
	TP_b + FN_b = 1 - y \\
	TN_b + FP_b = y \\
	FP_a = FP_b + \frac{x-y}{2} \\	
	FN_a = FN_b + \frac{y-x}{2} \\
	\frac{-FP_b}{2} = \frac{FN_b}{2} \lor x=y \\
	\frac{- FN_b}{2} = \frac{FP_b}{2} \lor x=y
\end{cases}
\Longleftrightarrow
\begin{cases}
	TP_a + FN_a = 1 - x \\
	TN_a + FP_a = x \\
	TP_b + FN_b = 1 - y \\
	TN_b + FP_b = y \\
	FP_a = FP_b + \frac{x-y}{2} \\	
	FN_a = FN_b + \frac{y-x}{2} \\
	-FP_b = FN_b = 0 \lor x=y \\
	- FN_b = FP_b = 0 \lor x=y
\end{cases}
$$
$$\Longleftrightarrow
\begin{cases}
	TP_a = 1 - x - \frac{y-x}{2} \\
	TN_a = x - \frac{x-y}{2} \\
	TP_b = 1 - y \\
	TN_b = y \\
	FP_a = \frac{x-y}{2} \\	
	FN_a = \frac{y-x}{2} \\
	-FP_b = FN_b = 0 
\end{cases} \lor
\begin{cases}
	TP_a = 1 - x - FN_a \\
	TN_a = x - FP_a \\
	TP_b = 1 - y - FN_b \\
	TN_b = y - FP_b \\
	FP_a = FP_b \\	
	FN_a = FN_b \\
	x=y 
\end{cases}
$$

\subsection{Equalised odds and Predictive Parity/Treatment Equality}
$$
\begin{cases}
	TP_a + FN_a = 1 - x \\
	TN_a + FP_a = x \\
	TP_b + FN_b = 1 - y \\
	TN_b + FP_b = y \\
	TP_a = \frac{1-x}{1-y} TP_b \\
	FP_a = \frac{x}{y} FP_b \\
	TP_a = \frac{1-x}{1-y} TP_b \\
	FP_a = \frac{1-x}{1-y} FP_b
\end{cases} \Longleftrightarrow
\begin{cases}
	TP_a + FN_a = 1 - x \\
	TN_a + FP_a = x \\
	TP_b + FN_b = 1 - y \\
	TN_b + FP_b = y \\
	TP_a = \frac{1-x}{1-y} TP_b \\
	\frac{1-x}{1-y} FP_b = \frac{x}{y} FP_b \\
	FP_a = \frac{1-x}{1-y} FP_b
\end{cases}
 \Longleftrightarrow
\begin{cases}
	TP_a + FN_a = 1 - x \\
	TN_a + FP_a = x \\
	TP_b + FN_b = 1 - y \\
	TN_b + FP_b = y \\
	TP_a = \frac{1-x}{1-y} TP_b \\
	\frac{1-x}{1-y} = \frac{x}{y} \lor FP_b = 0\\
	FP_a = \frac{1-x}{1-y} FP_b
\end{cases}$$
$$ \Longleftrightarrow
\begin{cases}
	TP_a + FN_a = 1 - x \\
	TN_a + FP_a = x \\
	TP_b + FN_b = 1 - y \\
	TN_b + FP_b = y \\
	TP_a = \frac{1-x}{1-y} TP_b \\
	y - xy = x - xy \lor FP_b = 0\\
	FP_a = \frac{1-x}{1-y} FP_b
\end{cases}
\Longleftrightarrow
\begin{cases}
	TP_a + FN_a = 1 - x \\
	TN_a + FP_a = x \\
	TP_b + FN_b = 1 - y \\
	TN_b + FP_b = y \\
	TP_a = \frac{1-x}{1-y} TP_b \\
	y = x \lor FP_b = 0\\
	FP_a = \frac{1-x}{1-y} FP_b
\end{cases}
 \Longleftrightarrow
\begin{cases}
	FN_a = 1 - x - TP_b \\
	TN_a = x - FP_b \\
	FN_b = 1 - y - TP_b\\
	TN_b = y - FP_b \\
	TP_a = TP_b \\
	y = x \\
	FP_a = FP_b
\end{cases} \lor
\begin{cases}
	FN_a = 1 - x - \frac{1-x}{1-y} TP_b \\
	TN_a = x \\
	FN_b = 1 - y - TP_b \\
	TN_b = y \\
	TP_a = \frac{1-x}{1-y} TP_b \\
	FP_b = 0\\
	FP_a = 0
\end{cases} 
$$

\subsection{Equalised odds and Predictive Equality}
$$
\begin{cases}
	TP_a + FN_a = 1 - x \\
	TN_a + FP_a = x \\
	TP_b + FN_b = 1 - y \\
	TN_b + FP_b = y \\
	FN_a = \frac{1-x}{1-y} FN_b \\
	TN_a = \frac{x}{y} TN_b \\
	FN_a = \frac{1-x}{1-y} FN_b \\
	TN_a = \frac{1-x}{1-y} TN_b
\end{cases} \Longleftrightarrow
\begin{cases}
	TP_a + FN_a = 1 - x \\
	TN_a + FP_a = x \\
	TP_b + FN_b = 1 - y \\
	TN_b + FP_b = y \\
	FN_a = \frac{1-x}{1-y} FN_b \\
	TN_a = \frac{x}{y} TN_b \\
	\frac{x}{y} TN_b = \frac{1-x}{1-y} TN_b
\end{cases} \Longleftrightarrow
\begin{cases}
	TP_a + FN_a = 1 - x \\
	TN_a + FP_a = x \\
	TP_b + FN_b = 1 - y \\
	TN_b + FP_b = y \\
	FN_a = \frac{1-x}{1-y} FN_b \\
	TN_a = \frac{x}{y} TN_b \\
	\frac{x}{y} = \frac{1-x}{1-y} \lor TN_b = 0
\end{cases}
$$
$$ \Longleftrightarrow
\begin{cases}
	TP_a + FN_a = 1 - x \\
	TN_a + FP_a = x \\
	TP_b + FN_b = 1 - y \\
	TN_b + FP_b = y \\
	FN_a = \frac{1-x}{1-y} FN_b \\
	TN_a = \frac{x}{y} TN_b \\
	x = y \lor TN_b = 0
\end{cases} \Longleftrightarrow
\begin{cases}
	TP_a = 1 - x - FN_b \\
	FP_a = x - TN_b \\
	TP_b = 1 - y - FN_b \\
	FP_b = y - TN_b \\
	FN_a = FN_b \\
	TN_a = TN_b \\
	x = y 
\end{cases}\lor 
\begin{cases}
	TP_a = 1 - x - \frac{1-x}{1-y} FN_b \\
	FP_a = x \\
	TP_b = 1 - y - FN_b \\
	FP_b = y \\
	FN_a = \frac{1-x}{1-y} FN_b \\
	TN_a = 0 \\
	TN_b = 0
\end{cases}
$$

\subsection{Equalised odds and Overall accuracy Equality}
$$
\begin{cases}
	TP_a + FN_a = 1 - x \\
	TN_a + FP_a = x \\
	TP_b + FN_b = 1 - y \\
	TN_b + FP_b = y \\
	TP_a = \frac{1-x}{1-y} TP_b \\
	FP_a = \frac{x}{y} FP_b \\
	TP_a = \frac{1-x}{1-y} TP_b \\
	\frac{y-x}{1-y}TP_b = TN_b - TN_a
\end{cases} \Longleftrightarrow
\begin{cases}
	TP_a + FN_a = 1 - x \\
	TN_a + FP_a = x \\
	TP_b + FN_b = 1 - y \\
	TN_b + FP_b = y \\
	TP_a = \frac{1-x}{1-y} TP_b \\
	FP_a = \frac{x}{y} FP_b \\
	\frac{1-x}{1-y} TP_b = \frac{1-x}{1-y} TP_b \\
	\frac{y-x}{1-y}TP_b = \frac{y-x}{y}TN_b
\end{cases} \Longleftrightarrow
\begin{cases}
	TP_a + FN_a = 1 - x \\
	TN_a + FP_a = x \\
	TP_b + FN_b = 1 - y \\
	TN_b + FP_b = y \\
	TP_a = \frac{1-x}{1-y} TP_b \\
	FP_a = \frac{x}{y} FP_b \\
	TP_b = \frac{1-y}{y}TN_b \lor x=y
\end{cases}
$$
$$ \Longleftrightarrow
\begin{cases}
	FN_a = 1 - x - \frac{1-x}{y} TN_b\\
	TN_a = x - \frac{x}{y} FP_b \\
	FN_b = 1 - y - \frac{1-y}{y}TN_b \\
	TN_b = y - FP_b\\
	TP_a = \frac{1-x}{y} TN_b \\
	FP_a = \frac{x}{y} FP_b \\
	TP_b = \frac{1-y}{y}TN_b
\end{cases} \lor
\begin{cases}
	FN_a = 1 - x - TP_b \\
	TN_a = x - FP_b\\
	FN_b = 1 - y - TP_b \\
	TN_b = y - FP_b \\
	TP_a = TP_b \\
	FP_a = FP_b \\
	x=y
\end{cases} 
$$

\subsection{Equal Opportunity, Predictive Equality and Overall accuracy Equality}
$$
\begin{cases}
	TP_a + FN_a = 1 - x \\
	TN_a + FP_a = x \\
	TP_b + FN_b = 1 - y \\
	TN_b + FP_b = y \\
	FN_a = \frac{1-x}{1-y} FN_b \\
	TN_a = \frac{1-x}{1-y} TN_b \\
	FN_a = \frac{1-x}{1-y} FN_b \\
	TN_a = TN_b - \frac{y-x}{1-y} TP_b       
\end{cases} \Longleftrightarrow
\begin{cases}
	TP_a + FN_a = 1 - x \\
	TN_a + FP_a = x \\
	TP_b + FN_b = 1 - y \\
	TN_b + FP_b = y \\
	FN_a = \frac{1-x}{1-y} FN_b \\
	TN_a = \frac{1-x}{1-y} TN_b \\
	\frac{1-x}{1-y} TN_b = TN_b - \frac{y-x}{1-y} TP_b       
\end{cases} \Longleftrightarrow
\begin{cases}
	TP_a + FN_a = 1 - x \\
	TN_a + FP_a = x \\
	TP_b + FN_b = 1 - y \\
	TN_b + FP_b = y \\
	FN_a = \frac{1-x}{1-y} FN_b \\
	TN_a = \frac{1-x}{1-y} TN_b \\
	\frac{1-x -1 + y}{1-y} TN_b = - \frac{y-x}{1-y} TP_b       
\end{cases}
$$
$$\Longleftrightarrow
\begin{cases}
	TP_a + FN_a = 1 - x \\
	TN_a + FP_a = x \\
	TP_b + FN_b = 1 - y \\
	TN_b + FP_b = y \\
	FN_a = \frac{1-x}{1-y} FN_b \\
	TN_a = \frac{1-x}{1-y} TN_b \\
	\frac{y-x}{1-y} TN_b = - \frac{y-x}{1-y} TP_b       
\end{cases} \Longleftrightarrow
\begin{cases}
	TP_a = 1 - x - FN_a \\
	FP_a = x - TN_a \\
	TP_b = 1 - y - FN_b\\
	FP_b = y - TN_b \\
	FN_a = \frac{1-x}{1-y} FN_b \\
	TN_a = \frac{1-x}{1-y} TN_b \\
	TN_b = - TP_b \lor x = y
\end{cases} \Longleftrightarrow
\begin{cases}
	TP_a = 1 - x - \frac{1-x}{1-y} FN_b \\
	FP_a = x - \frac{1-x}{1-y} TN_b \\
	TP_b = 1 - y - FN_b\\
	FP_b = y - TN_b \\
	FN_a = \frac{1-x}{1-y} FN_b \\
	TN_a = \frac{1-x}{1-y} TN_b \\
	TN_b = - TP_b = 0  \lor x = y
\end{cases}
$$

\subsection{Equal Opportunity, Predictive Equality and Predictive Parity/Treatment Equality}
$$
\begin{cases}
	TP_a + FN_a = 1 - x \\
	TN_a + FP_a = x \\
	TP_b + FN_b = 1 - y \\
	TN_b + FP_b = y \\
	TP_a = \frac{1-x}{1-y} TP_b \\
	FP_a = \frac{1-x}{1-y} FP_b \\
	TP_a = \frac{1-x}{1-y} TP_b \\
	TN_a = \frac{1-x}{1-y} TN_b
\end{cases} \Longleftrightarrow
\begin{cases}
	TP_a + FN_a = 1 - x \\
	TN_a + FP_a = x \\
	TP_b + FN_b = 1 - y \\
	TN_b + FP_b = y \\
	TP_a = \frac{1-x}{1-y} TP_b \\
	x - TN_a = \frac{1-x}{1-y} (y-TN_b) \\
	TN_a = \frac{1-x}{1-y} TN_b
\end{cases} \Longleftrightarrow
\begin{cases}
	TP_a + FN_a = 1 - x \\
	TN_a + FP_a = x \\
	TP_b + FN_b = 1 - y \\
	TN_b + FP_b = y \\
	TP_a = \frac{1-x}{1-y} TP_b \\
	TN_a = x - \frac{1-x}{1-y} (y-TN_b) \\
	x - \frac{1-x}{1-y} (y-TN_b) = \frac{1-x}{1-y} TN_b
\end{cases} 
$$
$$\Longleftrightarrow
\begin{cases}
	TP_a + FN_a = 1 - x \\
	TN_a + FP_a = x \\
	TP_b + FN_b = 1 - y \\
	TN_b + FP_b = y \\
	TP_a = \frac{1-x}{1-y} TP_b \\
	TN_a = x - \frac{1-x}{1-y} (y-TN_b) \\
	x = \frac{1-x}{1-y}y 
\end{cases} \Longleftrightarrow
\begin{cases}
	TP_a + FN_a = 1 - x \\
	TN_a + FP_a = x \\
	TP_b + FN_b = 1 - y \\
	TN_b + FP_b = y \\
	TP_a = \frac{1-x}{1-y} TP_b \\
	TN_a = x - \frac{1-x}{1-y} (y-TN_b) \\
	x=y
\end{cases} \Longleftrightarrow
\begin{cases}
	FN_a = 1 - x - TP_b \\
	FP_a = x - TN_b \\
	FN_b = 1 - y - TP_b \\
	FP_b = y - TN_b \\
	TP_a = TP_b \\
	TN_a = TN_b \\
	x = y 
\end{cases}
$$

\subsection{Equal Opportunity, Overall accuracy Equality and Predictive Parity/Treatment Equality}

$$
\begin{cases}
	TP_a + FN_a = 1 - x \\
	TN_a + FP_a = x \\
	TP_b + FN_b = 1 - y \\
	TN_b + FP_b = y \\
	TP_a = \frac{1-x}{1-y}TP_b \\
	FP_a = \frac{1-x}{1-y} FP_b \\
	TP_a = \frac{1-x}{1-y} TP_b \\
	TN_a = TN_b - \frac{y-x}{1-y} TP_b       
\end{cases} \Longleftrightarrow
\begin{cases}
	TP_a + FN_a = 1 - x \\
	TN_a + FP_a = x \\
	TP_b + FN_b = 1 - y \\
	TN_b + FP_b = y \\
	TP_a = \frac{1-x}{1-y}TP_b \\
	x-TN_a = \frac{1-x}{1-y} (y-TN_b) \\
	TN_a = TN_b - \frac{y-x}{1-y} TP_b          
\end{cases} 
$$
$$ \Longleftrightarrow
\begin{cases}
	TP_a + FN_a = 1 - x \\
	TN_a + FP_a = x \\
	TP_b + FN_b = 1 - y \\
	TN_b + FP_b = y \\
	TP_a = \frac{1-x}{1-y}TP_b \\
	TN_a = x - \frac{1-x}{1-y} (y-TN_b) \\
	x - \frac{1-x}{1-y} y + \frac{1-x}{1-y} TN_b = TN_b - \frac{y-x}{1-y} TP_b             
\end{cases} \Longleftrightarrow
\begin{cases}
	TP_a + FN_a = 1 - x \\
	TN_a + FP_a = x \\
	TP_b + FN_b = 1 - y \\
	TN_b + FP_b = y \\
	TP_a = \frac{1-x}{1-y}TP_b \\
	TN_a = x - \frac{1-x}{1-y} (y-TN_b) \\
	 - TN_b + \frac{1-x}{1-y} TN_b = - \frac{y-x}{1-y} TP_b  -x + \frac{1-x}{1-y} y 
\end{cases}
$$

$$ \Longleftrightarrow
\begin{cases}
	TP_a + FN_a = 1 - x \\
	TN_a + FP_a = x \\
	TP_b + FN_b = 1 - y \\
	TN_b + FP_b = y \\
	TP_a = \frac{1-x}{1-y}TP_b \\
	TN_a = x - \frac{1-x}{1-y} (y-TN_b) \\
	\frac{y-x}{1-y} TN_b = - \frac{y-x}{1-y} TP_b  -x + \frac{1-x}{1-y} y 
\end{cases} \Longleftrightarrow
\begin{cases}
	TP_a + FN_a = 1 - x \\
	TN_a + FP_a = x \\
	TP_b + FN_b = 1 - y \\
	TN_b + FP_b = y \\
	TP_a = \frac{1-x}{1-y}TP_b \\
	TN_a = x - \frac{1-x}{1-y} (y-TN_b) \\
	TN_b = - TP_b  -\frac{1-y}{y-x}x + \frac{1-x}{1-y} \frac{1-y}{y-x} y \lor x = y
\end{cases}
$$
$$ \Longleftrightarrow
\begin{cases}
	TP_a + FN_a = 1 - x \\
	TN_a + FP_a = x \\
	TP_b + FN_b = 1 - y \\
	TN_b + FP_b = y \\
	TP_a = \frac{1-x}{1-y}TP_b \\
	TN_a = x - \frac{1-x}{1-y} (y-TN_b) \\
	TN_b = - TP_b  -\frac{1-y}{y-x}x + \frac{1-x}{y-x} y \lor x = y
\end{cases} \Longleftrightarrow
\begin{cases}
	TP_a + FN_a = 1 - x \\
	TN_a + FP_a = x \\
	TP_b + FN_b = 1 - y \\
	TN_b + FP_b = y \\
	TP_a = \frac{1-x}{1-y}TP_b \\
	TN_a = x - \frac{1-x}{1-y} (y-TN_b) \\
	TN_b = - TP_b + \frac{1}{y-x} (-x + y) \lor x = y
\end{cases} 
$$
$$ \Longleftrightarrow
\begin{cases}
	FN_a = 1 - x - TP_a \\
	TN_a = x - FP_a \\
	FN_b = 1 - y - TP_b \\
	TN_b = y - FP_b\\
	TP_a = \frac{1-x}{1-y}TP_b \\
	TN_a = 1 - \frac{1-x}{1-y} TP_b \\
	TN_b =  1 - TP_b
\end{cases} \lor
\begin{cases}
	FN_a = 1 - x - TP_a \\
	TN_a = x - FP_a \\
	FN_b = 1 - y - TP_b \\
	TN_b = y - FP_b\\
	TP_a = TP_b \\
	TN_a = TN_b \\
	x = y
\end{cases}
$$
$$ \Longleftrightarrow
\begin{cases}
	FN_a = 0 \\
	FP_a = 0 \\
	FN_b = 0 \\
	TN_b = 0\\
	TP_a = 1 - x \\
	TN_a = x \\
	TN_b =  y \\
	TP_b = 1 - y
\end{cases} \lor
\begin{cases}
	FN_a = 1 - x - TP_b \\
	TN_a = x - FP_b \\
	FN_b = 1 - y - TP_b \\
	TN_b = y - FP_b\\
	TP_a = TP_b \\
	TN_a = TN_b \\
	x = y
\end{cases}
$$

\subsection{False positive parity, Predictive Parity and  Predictive Equality/Treatment Equality}
$$
\begin{cases}
	TP_a + FN_a = 1 - x \\
	TN_a + FP_a = x \\
	TP_b + FN_b = 1 - y \\
	TN_b + FP_b = y \\
	FP_a = \frac{x}{y} FP_b \\
	TP_a = \frac{x}{y} TP_b \\
	FP_a = \frac{x}{y} FP_b \\
	FN_a = \frac{x}{y} FN_b
\end{cases} \Longleftrightarrow
\begin{cases}
	TP_a + FN_a = 1 - x \\
	TN_a + FP_a = x \\
	TP_b + FN_b = 1 - y \\
	TN_b + FP_b = y \\
	FP_a = \frac{x}{y} FP_b \\
	FN_a = \frac{x}{y} FN_b \\
	1 - x - FN_a = \frac{x}{y} ( 1 - y - FN_b)
\end{cases}
$$
$$ \Longleftrightarrow
\begin{cases}
	TP_a + FN_a = 1 - x \\
	TN_a + FP_a = x \\
	TP_b + FN_b = 1 - y \\
	TN_b + FP_b = y \\
	FP_a = \frac{x}{y} FP_b \\
	FN_a = \frac{x}{y} FN_b \\
	1 - x - \frac{x}{y} FN_b = \frac{x}{y} - x - \frac{x}{y} FN_b
\end{cases} \Longleftrightarrow
\begin{cases}
	TP_a + FN_a = 1 - x \\
	TN_a + FP_a = x \\
	TP_b + FN_b = 1 - y \\
	TN_b + FP_b = y \\
	FP_a = \frac{x}{y} FP_b \\
	FN_a = \frac{x}{y} FN_b \\
	1 = \frac{x}{y}
\end{cases} \Longleftrightarrow
\begin{cases}
	TP_a = 1 - x - FN_b \\
	TN_a = x - FP_b \\
	TP_b = 1 - y - FN_b \\
	TN_b = y - FP_b \\
	y = x \\
	FP_a = FP_b \\
	FN_a = FN_b 
\end{cases}
$$

\subsection{False positive parity, Predictive Parity and Overall accuracy Equality}
$$
\begin{cases}
	TP_a + FN_a = 1 - x \\
	TN_a + FP_a = x \\
	TP_b + FN_b = 1 - y \\
	TN_b + FP_b = y \\
	TN_a = \frac{x}{y} TN_b \\
	TP_a = \frac{x}{y} TP_b \\
	TN_a = \frac{x}{x-y}(TP_b-TP_a) \\
	TN_b = \frac{y}{x-y} (TP_b - TP_a)
\end{cases} \Longleftrightarrow
\begin{cases}
	TP_a + FN_a = 1 - x \\
	TN_a + FP_a = x \\
	TP_b + FN_b = 1 - y \\
	TN_b + FP_b = y \\
	TN_a = \frac{x}{y} TN_b \\
	TP_a = \frac{x}{y} TP_b \\
	TN_a = \frac{x}{x-y}(TP_b-\frac{x}{y} TP_b) \\
	TN_b = \frac{y}{x-y} (TP_b - \frac{x}{y} TP_b)
\end{cases}
\Longleftrightarrow
\begin{cases}
	TP_a + FN_a = 1 - x \\
	TN_a + FP_a = x \\
	TP_b + FN_b = 1 - y \\
	TN_b + FP_b = y \\
	TN_a = \frac{x}{y} TN_b \\
	TP_a = \frac{x}{y} TP_b \\
	TN_a = \frac{x}{x-y}(\frac{y-x}{y} TP_b) \\
	TN_b = \frac{y}{x-y} (\frac{y-x}{y} TP_b)
\end{cases}
$$

$$\Longleftrightarrow
\begin{cases}
	TP_a + FN_a = 1 - x \\
	TN_a + FP_a = x \\
	TP_b + FN_b = 1 - y \\
	TN_b + FP_b = y \\
	TN_a = \frac{x}{y} TN_b \\
	TP_a = \frac{x}{y} TP_b \\
	TN_a = -\frac{x}{y}TP_b \\
	TN_b = -TP_b
\end{cases} \Longleftrightarrow
\begin{cases}
	TP_a + FN_a = 1 - x \\
	TN_a + FP_a = x \\
	TP_b + FN_b = 1 - y \\
	TN_b + FP_b = y \\
	TP_a = \frac{x}{y} TP_b \\
	TN_a = -\frac{x}{y}TP_b \\
	TN_b = - TP_b \\
	-\frac{x}{y}TP_b = -\frac{x}{y} TP_b \\
\end{cases} \Longleftrightarrow
\begin{cases}
	FN_a = 1 - x\\
	FP_a = x \\
	FN_b = 1 - y \\
	FP_b = y \\
	TP_a = \frac{x}{y} TP_b = 0 \\
	TN_a = -\frac{x}{y}TP_b  = 0\\
	TN_b = -TP_b = 0\\
	1 = 1 \\
\end{cases}
$$

\subsection{False positive parity, Overall accuracy Equality and Predictive Equality/Treatment equality}
$$
\begin{cases}
	TP_a + FN_a = 1 - x \\
	TN_a + FP_a = x \\
	TP_b + FN_b = 1 - y \\
	TN_b + FP_b = y \\
	TN_a = \frac{x}{y} TN_b \\
	FN_a = \frac{x}{y} FN_b \\
	TN_a = \frac{x}{x-y}(TP_b-TP_a) \\
	TN_b = \frac{y}{x-y} (TP_b - TP_a)
\end{cases} \Longleftrightarrow
\begin{cases}
	TP_a + FN_a = 1 - x \\
	TN_a + FP_a = x \\
	TP_b + FN_b = 1 - y \\
	TN_b + FP_b = y \\
	TN_a = \frac{x}{y} TN_b \\
	TP_a = 1 -\frac{x}{y} + \frac{x}{y}TP_b \\
	TN_a = \frac{x}{x-y}(TP_b-1 + \frac{x}{y} - \frac{x}{y}TP_b) \\
	TN_b = \frac{y}{x-y} (TP_b - 1 + \frac{x}{y} - \frac{x}{y}TP_b)
\end{cases}
$$
$$ \Longleftrightarrow
\begin{cases}
	TP_a + FN_a = 1 - x \\
	TN_a + FP_a = x \\
	TP_b + FN_b = 1 - y \\
	TN_b + FP_b = y \\
	TN_a = \frac{x}{y} TN_b \\
	TP_a = 1 -\frac{x}{y} + \frac{x}{y}TP_b \\
	TN_a = \frac{x}{x-y}(\frac{x - y}{y} + \frac{y-x}{y}TP_b) \\
	TN_b = \frac{y}{x-y} (\frac{x - y}{y} + \frac{y-x}{y}TP_b)
\end{cases} \Longleftrightarrow
\begin{cases}
	TP_a + FN_a = 1 - x \\
	TN_a + FP_a = x \\
	TP_b + FN_b = 1 - y \\
	TN_b + FP_b = y \\
	TN_a = \frac{x}{y} TN_b \\
	TP_a = 1 -\frac{x}{y} + \frac{x}{y}TP_b \\
	TN_a = \frac{x}{y} - \frac{x}{y}TP_b \\
	TN_b = 1 - TP_b
\end{cases}
$$
$$\Longleftrightarrow
\begin{cases}
	TP_a + FN_a = 1 - x \\
	TN_a + FP_a = x \\
	TP_b + FN_b = 1 - y \\
	TN_b + FP_b = y \\
	TP_a = 1 -\frac{x}{y} + \frac{x}{y}TP_b \\
	TN_a = \frac{x}{y} - \frac{x}{y}TP_b \\
	TN_b = 1 - TP_b \\
	\frac{x}{y} - \frac{x}{y}TP_b = \frac{x}{y} (1 - TP_b) \\
\end{cases} \Longleftrightarrow
\begin{cases}
	FN_a = 1 - x - TP_a \\
	FP_a = x - TN_a \\
	FN_b = 1 - y - TP_b \\
	FP_b = y - TN_b \\
	TP_a = 1 -\frac{x}{y} + \frac{x}{y}TP_b \\
	TN_a = \frac{x}{y} - \frac{x}{y}TP_b \\
	TN_b = 1 - TP_b \\
	1 = 1 \\
\end{cases}\Longleftrightarrow
\begin{cases}
	FN_a = 0 \\
	FP_a = 0 \\
	FN_b = 0 \\
	FP_b = 0 \\
	TP_a = 1 -x \\
	TN_a = x\\
	TN_b = y \\
	TP_b = 1 - y \\
	1 = 1 \\
\end{cases}
$$

\subsection{Predictive Parity, Predictive Equality and Overall accuracy Equality}
\label{subsec:pp_pe_oae}
$$
\begin{cases}
	TP_a + FN_a = 1 - x \\
	TN_a + FP_a = x \\
	TP_b + FN_b = 1 - y \\
	TN_b + FP_b = y \\
	TP_a = \frac{TP_b * FP_a}{FP_b} \\
	FN_a = \frac{FN_b*TN_a}{TN_b} \\
	TP_a + TN_a = TP_b + TN_b
\end{cases} \Longleftrightarrow
\begin{cases}
	TP_a + FN_a = 1 - x \\
	TN_a + FP_a = x \\
	TP_b + FN_b = 1 - y \\
	TN_b + FP_b = y \\
	TP_a*(y-TN_b) = TP_b * (x-TN_a) \\
	(1-x-TP_a)*TN_b = (1-y-TP_b)*TN_a \\
	TP_a + TN_a = TP_b + TN_b
\end{cases}
$$
$$ \Longleftrightarrow
\begin{cases}
	TP_a + FN_a = 1 - x \\
	TN_a + FP_a = x \\
	TP_b + FN_b = 1 - y \\
	TN_b + FP_b = y \\
	TP_a= TP_b * \frac{x-TN_a}{y-TN_b} \\
	(1-x-TP_b * \frac{x-TN_a}{y-TN_b})*TN_b = (1-y-TP_b)*TN_a \\
	TP_a + TN_a = TP_b + TN_b
\end{cases}
$$

$$ \Longleftrightarrow
\begin{cases}
	TP_a + FN_a = 1 - x \\
	TN_a + FP_a = x \\
	TP_b + FN_b = 1 - y \\
	TN_b + FP_b = y \\
	TP_a= TP_b * \frac{x-TN_a}{y-TN_b} \\
	-TP_b TN_b \frac{x-TN_a}{y-TN_b} + TP_b TN_a = -TN_b + xTN_b + TN_a - yTN_a\\
	TP_a + TN_a = TP_b + TN_b
\end{cases}
$$
$$ \Longleftrightarrow
\begin{cases}
	TP_a + FN_a = 1 - x \\
	TN_a + FP_a = x \\
	TP_b + FN_b = 1 - y \\
	TN_b + FP_b = y \\
	TP_a= TP_b * \frac{x-TN_a}{y-TN_b} \\
	TP_b(TN_a - TN_b* \frac{x-TN_a}{y-TN_b}) = (x-1)TN_b + (1-y)TN_a\\
	TP_a + TN_a = TP_b + TN_b
\end{cases}
$$
$$ \Longleftrightarrow
\begin{cases}
	TP_a + FN_a = 1 - x \\
	TN_a + FP_a = x \\
	TP_b + FN_b = 1 - y \\
	TN_b + FP_b = y \\
	TP_a= TP_b * \frac{x-TN_a}{y-TN_b} \\
	TP_b(\frac{yTN_a - TN_aTN_b-xTN_b+TN_bTN_a}{y-TN_b}) = (x-1)TN_b + (1-y)TN_a\\
	TP_a + TN_a = TP_b + TN_b
\end{cases}
$$
$$ \Longleftrightarrow
\begin{cases}
	TP_a + FN_a = 1 - x \\
	TN_a + FP_a = x \\
	TP_b + FN_b = 1 - y \\
	TN_b + FP_b = y \\
	TP_a= TP_b * \frac{x-TN_a}{y-TN_b} \\
	TP_b(\frac{yTN_a-xTN_b}{y-TN_b}) = (x-1)TN_b + (1-y)TN_a\\
	TP_a + TN_a = TP_b + TN_b
\end{cases} \Longleftrightarrow
\begin{cases}
	TP_a + FN_a = 1 - x \\
	TN_a + FP_a = x \\
	TP_b + FN_b = 1 - y \\
	TN_b + FP_b = y \\
	TP_a= \frac{(x-1)(x-TN_a)}{yTN_a-xTN_b}TN_b + \frac{(1-y)(x-TN_a)}{yTN_a-xTN_b}TN_a\\
	TP_b = \frac{(x-1)(y-TN_b)}{yTN_a-xTN_b}TN_b + \frac{(1-y)(y-TN_b)}{yTN_a-xTN_b}TN_a\\
	TP_a + TN_a = TP_b + TN_b
\end{cases}
$$
$$ \Longleftrightarrow
\begin{cases}
	TP_a + FN_a = 1 - x \\
	TN_a + FP_a = x \\
	TP_b + FN_b = 1 - y \\
	TN_b + FP_b = y \\
	TP_a= \frac{(x-1)(x-TN_a)}{yTN_a-xTN_b}TN_b + \frac{(1-y)(x-TN_a)}{yTN_a-xTN_b}TN_a\\
	TP_b = \frac{(x-1)(y-TN_b)}{yTN_a-xTN_b}TN_b + \frac{(1-y)(y-TN_b)}{yTN_a-xTN_b}TN_a\\
	\frac{(x-1)(x-TN_a)}{yTN_a-xTN_b}TN_b + \frac{(1-y)(x-TN_a)}{yTN_a-xTN_b}TN_a + TN_a = \frac{(x-1)(y-TN_b)}{yTN_a-xTN_b}TN_b + \frac{(1-y)(y-TN_b)}{yTN_a-xTN_b}TN_a + TN_b
\end{cases}
$$
$$ \Longleftrightarrow
\begin{cases}
	TP_a + FN_a = 1 - x \\
	TN_a + FP_a = x \\
	TP_b + FN_b = 1 - y \\
	TN_b + FP_b = y \\
	TP_a= \frac{(1-x)(x-TN_a)}{yTN_a-xTN_b}TN_b + \frac{(1-y)(x-TN_a)}{yTN_a-xTN_b}TN_a\\
	TP_b = \frac{(1-x)(y-TN_b)}{yTN_a-xTN_b}TN_b + \frac{(1-y)(y-TN_b)}{yTN_a-xTN_b}TN_a\\

	(x^2-xTN_a-x+TN_a)TN_b + (x-xy-TN_a+yTN_a)TN_a + (yTN_a-xTN_b)TN_a = \\(xy - xTN_b-y+<tn-b)TN_b + (y-TN_b-y^2+yTN_b)TN_a + (yTN_a-xTN_b)TN_b	
	
\end{cases}
$$

$$ \Longleftrightarrow
\begin{cases}
	TP_a + FN_a = 1 - x \\
	TN_a + FP_a = x \\
	TP_b + FN_b = 1 - y \\
	TN_b + FP_b = y \\
	TP_a= \frac{(1-x)(x-TN_a)}{yTN_a-xTN_b}TN_b + \frac{(1-y)(x-TN_a)}{yTN_a-xTN_b}TN_a\\
	TP_b = \frac{(1-x)(y-TN_b)}{yTN_a-xTN_b}TN_b + \frac{(1-y)(y-TN_b)}{yTN_a-xTN_b}TN_a\\

	(xy - 2xTN_b - y + TN_b + yTN_a - x^2 +xTN_a + x - TN_a)TN_b = \\
	(x - xy - TN_a + 2yTN_a - xTN_b - y +TN_b + y ^2 -y TN_b)TN_a
	
\end{cases}
$$

$$ \Longleftrightarrow
\begin{cases}
	TP_a + FN_a = 1 - x \\
	TN_a + FP_a = x \\
	TP_b + FN_b = 1 - y \\
	TN_b + FP_b = y \\
	TP_a= \frac{(1-x)(x-TN_a)}{yTN_a-xTN_b}TN_b + \frac{(1-y)(x-TN_a)}{yTN_a-xTN_b}TN_a\\
	TP_b = \frac{(1-x)(y-TN_b)}{yTN_a-xTN_b}TN_b + \frac{(1-y)(y-TN_b)}{yTN_a-xTN_b}TN_a\\

	(-2x+1)TN_b^2 + (xy - y + x + 2yTN_a-x^2 +2xTN_a-2TN_a)TN_b + (xy - x + TN_a -2yTN_a + y - y^2)TN_a = 0
	
\end{cases}
$$

\subsection{Predictive Parity, Predictive Equality and Treatment Equality}
$$
\begin{cases}
	TP_a + FN_a = 1 - x \\
	TN_a + FP_a = x \\
	TP_b + FN_b = 1 - y \\
	TN_b + FP_b = y \\
	TP_a = \frac{TP_bFP_a}{FP_b} \\
	FN_a = \frac{FN_bTN_a}{TN_b} \\
	FN_a = \frac{1-x}{1-y} FN_b \\
	FP_a = \frac{1-x}{1-y} FP_b
\end{cases} \Longleftrightarrow
\begin{cases}
	TP_a + FN_a = 1 - x \\
	TN_a + FP_a = x \\
	TP_b + FN_b = 1 - y \\
	TN_b + FP_b = y \\
	\frac{1-x}{1-y} TP_b = \frac{TP_bFP_a}{FP_b} \\
	\frac{1-x}{1-y} FN_b = \frac{FN_bTN_a}{TN_b} \\
	FN_a = \frac{1-x}{1-y} FN_b \\
	FP_a = \frac{1-x}{1-y} FP_b
\end{cases} \Longleftrightarrow
\begin{cases}
	TP_a + FN_a = 1 - x \\
	TN_a + FP_a = x \\
	TP_b + FN_b = 1 - y \\
	TN_b + FP_b = y \\
	\frac{1-x}{1-y} FP_b = FP_a \lor TP_b = 0 \\
	TN_a = \frac{1-x}{1-y} TN_b \lor FN_b = 0 \\
	FN_a = \frac{1-x}{1-y} FN_b \\
	FP_a = \frac{1-x}{1-y} FP_b
\end{cases}
$$
$$\Longleftrightarrow
\begin{cases}
	TP_a + FN_a = 1 - x \\
	TN_a + FP_a = x \\
	TP_b + FN_b = 1 - y \\
	TN_b + FP_b = y \\
	TN_a = \frac{1-x}{1-y} TN_b \lor FN_b = 0 \\ 
	FN_a = \frac{1-x}{1-y} FN_b \\
	TN_a = x - \frac{1-x}{1-y} y + \frac{1-x}{1-y} TN_b
\end{cases} \Longleftrightarrow
\begin{cases}
	TP_a + FN_a = 1 - x \\
	TN_a + FP_a = x \\
	TP_b + FN_b = 1 - y \\
	TN_b + FP_b = y \\
	TN_a = x - \frac{1-x}{1-y} y + \frac{1-x}{1-y} TN_b \\
	TN_a = \frac{1-x}{1-y} TN_b \\
	FN_a = \frac{1-x}{1-y} FN_b
\end{cases} \lor 
\begin{cases}
	TP_a = 1 - x \\
	FP_a = x - TN_a \\
	TP_b = 1 - y \\
	FP_b = y - TN_b \\
	TN_a = x - \frac{1-x}{1-y} y + \frac{1-x}{1-y} TN_b \\
	FN_b = 0 \\
	FN_a = 0
\end{cases}
$$
$$
\Longleftrightarrow
\begin{cases}
	TP_a + FN_a = 1 - x \\
	TN_a + FP_a = x \\
	TP_b + FN_b = 1 - y \\
	TN_b + FP_b = y \\
	\frac{1-x}{1-y} TN_b  = x - \frac{1-x}{1-y} y + \frac{1-x}{1-y} TN_b \\
	TN_a = \frac{1-x}{1-y} TN_b \\
	FN_a = \frac{1-x}{1-y} FN_b
\end{cases} \lor 
\begin{cases}
	TP_a = 1 - x \\
	FP_a = x - TN_a \\
	TP_b = 1 - y \\
	FP_b = y - TN_b \\
	TN_a = x - \frac{1-x}{1-y} y + \frac{1-x}{1-y} TN_b \\
	FN_b = 0 \\
	FN_a = 0
\end{cases}
$$
$$
\Longleftrightarrow
\begin{cases}
	TP_a = 1 - x - FN_b \\
	FP_a = x - TN_b\\
	TP_b = 1 - y - FN_b \\
	FP_b = y - TN_b \\
	x = y  \\
	TN_a = TN_b \\
	FN_a = FN_b
\end{cases} \lor 
\begin{cases}
	TP_a = 1 - x \\
	FP_a = \frac{1-x}{1-y} y - \frac{1-x}{1-y} TN_b \\
	TP_b = 1 - y \\
	FP_b = y - TN_b \\
	TN_a = x - \frac{1-x}{1-y} y + \frac{1-x}{1-y} TN_b \\
	FN_b = 0 \\
	FN_a = 0
\end{cases}
$$

\subsection{Predictive Parity, Overall accuracy Equality and Treatment Equality}
$$
\begin{cases}
	TP_a + FN_a = 1 - x \\
	TN_a + FP_a = x \\
	TP_b + FN_b = 1 - y \\
	TN_b + FP_b = y \\
	FP_a = FP_b \\
	FN_a = FN_b \\
	TP_a = \frac{1-x}{1-y}TP_b \\
	FP_a = \frac{1-x}{1-y} FP_b
\end{cases} \Longleftrightarrow
\begin{cases}
	TP_a + FN_a = 1 - x \\
	TN_a + FP_a = x \\
	TP_b + FN_b = 1 - y \\
	TN_b + FP_b = y \\
	FP_a = FP_b \\
	1 - x - TP_a = 1 - y - TP_b\\
	TP_a = \frac{1-x}{1-y}TP_b \\
	FP_b = \frac{1-x}{1-y} FP_b
\end{cases} \Longleftrightarrow
\begin{cases}
	TP_a + FN_a = 1 - x \\
	TN_a + FP_a = x \\
	TP_b + FN_b = 1 - y \\
	TN_b + FP_b = y \\
	FP_a = FP_b \\
	TP_a = y - x + TP_b\\
	y - x - TP_b = \frac{1-x}{1-y}TP_b \\
	x = y \lor FP_b = 0
\end{cases}
$$
$$ \Longleftrightarrow
\begin{cases}
	TP_a + FN_a = 1 - x \\
	TN_a + FP_a = x \\
	TP_b + FN_b = 1 - y \\
	TN_b + FP_b = y \\
	FP_a = FP_b \\
	TP_a = y - x + TP_b\\
	y - x  = \frac{y-x}{1-y}TP_b \\
	x = y \lor FP_b = 0
\end{cases} \Longleftrightarrow
\begin{cases}
	TP_a + FN_a = 1 - x \\
	TN_a + FP_a = x \\
	TP_b + FN_b = 1 - y \\
	TN_b + FP_b = y \\
	FP_a = FP_b \\
	TP_a = y - x + TP_b\\
	TP_b = 1 - y \lor x=y \\
	x = y \lor FP_b = 0
\end{cases}
\Longleftrightarrow
\begin{cases}
	FN_a = 0 \\
	TN_a = x \\
	FN_b = 0 \\
	TN_b = y \\
	FP_a = 0 \\
	TP_a = 1 - x\\
	TP_b = 1 - y \\
	FP_b = 0
\end{cases} \lor
\begin{cases}
	FN_a = 1 - x - TP_b \\
	TN_a = x - FP_b\\
	FN_b = 1 - y - TP_b \\
	TN_b = y - FP_b \\
	FP_a = FP_b \\
	TP_a = TP_b\\
	x=y 
\end{cases}
$$

\subsection{Predictive Equality, Overall accuracy Equality and Treatment Equality}
$$
\begin{cases}
	TP_a + FN_a = 1 - x \\
	TN_a + FP_a = x \\
	TP_b + FN_b = 1 - y \\
	TN_b + FP_b = y \\
	FP_a = FP_b \\
	FN_a = FN_b \\
	TN_a = \frac{x}{y} TN_b \\
	FN_a = \frac{x}{y} FN_b
\end{cases} \Longleftrightarrow
\begin{cases}
	TP_a + FN_a = 1 - x \\
	TN_a + FP_a = x \\
	TP_b + FN_b = 1 - y \\
	TN_b + FP_b = y \\
	x - TN_a = y - TN_b \\
	FN_a = FN_b \\
	TN_a = \frac{x}{y} TN_b \\
	FN_b = \frac{x}{y} FN_b
\end{cases} \Longleftrightarrow
\begin{cases}
	TP_a + FN_a = 1 - x \\
	TN_a + FP_a = x \\
	TP_b + FN_b = 1 - y \\
	TN_b + FP_b = y \\
	TN_a =  x - y + TN_b \\
	FN_a = FN_b \\
	x - y + TN_b = \frac{x}{y} TN_b \\
	1 = \frac{x}{y} \lor FN_b = 0
\end{cases}
$$
$$ \Longleftrightarrow
\begin{cases}
	TP_a + FN_a = 1 - x \\
	TN_a + FP_a = x \\
	TP_b + FN_b = 1 - y \\
	TN_b + FP_b = y \\
	TN_a =  x - y + TN_b \\
	FN_a = FN_b \\
	x - y = \frac{x-y}{y} TN_b \\
	x = y \lor FN_b = 0
\end{cases} \Longleftrightarrow
\begin{cases}
	TP_a + FN_a = 1 - x \\
	TN_a + FP_a = x \\
	TP_b + FN_b = 1 - y \\
	TN_b + FP_b = y \\
	TN_a =  x - y + TN_b \\
	FN_a = FN_b \\
	1 = \frac{1}{y} TN_b \lor x=y \\
	x = y \lor FN_b = 0 
\end{cases} \Longleftrightarrow
\begin{cases}
	TP_a + FN_a = 1 - x \\
	TN_a + FP_a = x \\
	TP_b + FN_b = 1 - y \\
	TN_b + FP_b = y \\
	TN_a =  x - y + TN_b \\
	FN_a = FN_b \\
	TN_b = y \lor x=y \\
	x = y \lor FN_b = 0
\end{cases}
$$
$$ \Longleftrightarrow
\begin{cases}
	TP_a = 1 - x \\
	FP_a = 0 \\
	TP_b = 1 - y \\
	FP_b = 0 \\
	TN_a = x \\
	FN_a = 0 \\
	TN_b = y \\
	FN_b = 0
\end{cases} \lor
\begin{cases}
	TP_a = 1 - x - FN_b \\
	FP_a = x - TN_b \\
	TP_b = 1 - y - FN_b  \\
	FP_b = y - TN_b  \\
	TN_a = TN_b \\
	FN_a = FN_b \\
	x = y 
\end{cases}
$$

\end{document}